\documentclass[10pt,onecolumn]{IEEEtran}
\usepackage{soul}
\usepackage[mathscr]{eucal}
\usepackage[cmex10]{amsmath}
\usepackage{epsfig,epsf,psfrag}
\usepackage{amssymb,amsmath,amsthm,amsfonts,latexsym}
\usepackage{amsmath,graphicx,bm,xcolor,url,overpic}
\usepackage{fixltx2e}%ordering of single and double column floats
\usepackage{array}%array and tabular environments
\usepackage{verbatim}
\usepackage{bm}
\usepackage{algorithmic}
\usepackage{algorithm}
\usepackage{verbatim}
\usepackage{textcomp}
\usepackage{mathrsfs}
\usepackage{epstopdf}

\newcommand{\openone}{\leavevmode\hbox{\small1\normalsize\kern-.33em1}}

\catcode`~=11 \def\UrlSpecials{\do\~{\kern -.15em\lower .7ex\hbox{~}\kern .04em}} \catcode`~=13

\allowdisplaybreaks[4]

\newcommand{\nn}{\nonumber}

\newcommand{\calA}{\mathcal{A}}
\newcommand{\calB}{\mathcal{B}}
\newcommand{\calC}{\mathcal{C}}
\newcommand{\calD}{\mathcal{D}}

\newcommand{\calH}{\mathcal{H}}

\newcommand{\calK}{\mathcal{K}}
\newcommand{\calL}{\mathcal{L}}
\newcommand{\calM}{\mathcal{M}}
\newcommand{\calN}{\mathcal{N}}

\newcommand{\calP}{\mathcal{P}}
\newcommand{\calQ}{\mathcal{Q}}
\newcommand{\calR}{\mathcal{R}}
\newcommand{\calS}{\mathcal{S}}
\newcommand{\calT}{\mathcal{T}}
\newcommand{\calU}{\mathcal{U}}
\newcommand{\calV}{\mathcal{V}}
\newcommand{\calW}{\mathcal{W}}
\newcommand{\calX}{\mathcal{X}}
\newcommand{\calY}{\mathcal{Y}}

\newcommand{\hatcalX}{\hat{\calX}}

\newcommand{\rmb}{\mathrm{b}}

\newcommand{\rme}{\mathrm{e}}

\newcommand{\rmL}{\mathrm{L}}

\newcommand{\rmP}{\mathrm{P}}

\newcommand{\bbN}{\mathbb{N}}

\newcommand{\bbR}{\mathbb{R}}

\DeclareMathAlphabet{\mathbsf}{OT1}{cmss}{bx}{n}
\DeclareMathAlphabet{\mathssf}{OT1}{cmss}{m}{sl}% slanted sans serif

\DeclareSymbolFont{bsfletters}{OT1}{cmss}{bx}{n}
\DeclareSymbolFont{ssfletters}{OT1}{cmss}{m}{n}
\DeclareMathSymbol{\bsfGamma}{0}{bsfletters}{'000}
\DeclareMathSymbol{\ssfGamma}{0}{ssfletters}{'000}
\DeclareMathSymbol{\bsfDelta}{0}{bsfletters}{'001}
\DeclareMathSymbol{\ssfDelta}{0}{ssfletters}{'001}
\DeclareMathSymbol{\bsfTheta}{0}{bsfletters}{'002}
\DeclareMathSymbol{\ssfTheta}{0}{ssfletters}{'002}
\DeclareMathSymbol{\bsfLambda}{0}{bsfletters}{'003}
\DeclareMathSymbol{\ssfLambda}{0}{ssfletters}{'003}
\DeclareMathSymbol{\bsfXi}{0}{bsfletters}{'004}
\DeclareMathSymbol{\ssfXi}{0}{ssfletters}{'004}
\DeclareMathSymbol{\bsfPi}{0}{bsfletters}{'005}
\DeclareMathSymbol{\ssfPi}{0}{ssfletters}{'005}
\DeclareMathSymbol{\bsfSigma}{0}{bsfletters}{'006}
\DeclareMathSymbol{\ssfSigma}{0}{ssfletters}{'006}
\DeclareMathSymbol{\bsfUpsilon}{0}{bsfletters}{'007}
\DeclareMathSymbol{\ssfUpsilon}{0}{ssfletters}{'007}
\DeclareMathSymbol{\bsfPhi}{0}{bsfletters}{'010}
\DeclareMathSymbol{\ssfPhi}{0}{ssfletters}{'010}
\DeclareMathSymbol{\bsfPsi}{0}{bsfletters}{'011}
\DeclareMathSymbol{\ssfPsi}{0}{ssfletters}{'011}
\DeclareMathSymbol{\bsfOmega}{0}{bsfletters}{'012}
\DeclareMathSymbol{\ssfOmega}{0}{ssfletters}{'012}

\newcommand{\hatU}{\hat{U}}

\newcommand{\hatx}{\hat{x}}
\newcommand{\hatX}{\hat{X}}

\newtheorem{theorem}{Theorem}
\newtheorem{lemma}[theorem]{Lemma}

\newtheorem{corollary}[theorem]{Corollary}
\newtheorem{definition}{Definition}

\newtheorem{remark}{Remark}

\graphicspath{{./figures/}}
\usepackage{graphicx,cite}
\usepackage{epstopdf}
\usepackage{enumerate}
\usepackage[ colorlinks = true,
             linkcolor = blue,
             urlcolor  = blue,
             citecolor = red,
             anchorcolor = green,]{hyperref}
\usepackage{soul,color}
\usepackage{multirow}
\IEEEoverridecommandlockouts
\soulregister\em7
\soulregister\cite7
\soulregister\ref7
\soulregister\eqref7
\soulregister\underline7
\soulregister\emph7
\soulregister\footnote7
\definecolor{Dyellow}{RGB}{254,152,0}
\definecolor{Dgreen}{RGB}{0,176,80}

\begin{document}

\title{Successive Refinement of Shannon Cipher System Under Maximal Leakage}

\author{
\IEEEauthorblockN{Zhuangfei Wu, Lin Bai and Lin Zhou}
%\IEEEauthorblockA{School of Cyber Science and Technology (CST)\\
%Beihang University, China}
%\thanks{This work was supported in part by the National Key Research and Development Program of China under Grant 2020YFB1804800, in part by the National Natural Science Foundation of China under Grant U22B2008 and Grant 62201022, in part by the Beijing Natural Science Foundation under Grant JQ20019 and 4232007, in part by funds of Beihang University.}
\thanks{This work has been partially presented at ISIT 2023~\cite{wu2023isit}.}
\thanks{The authors are with the School of Cyber Science and Technology, Beihang University, Beijing, China, 100191 (Emails: \{zhuangfeiwu, l.bai, lzhou\}@buaa.edu.cn). }
}

\maketitle
\flushbottom
\allowdisplaybreaks[1]

\begin{abstract}
We study the successive refinement setting of Shannon cipher system (SCS) under the maximal leakage secrecy metric for discrete memoryless sources under bounded distortion measures. Specifically, we generalize the threat model for the point-to-point rate-distortion setting of Issa, Wagner and Kamath (T-IT 2020) to the multiterminal successive refinement setting. Under mild conditions that correspond to partial secrecy, we characterize the asymptotically optimal normalized maximal leakage region for both the joint excess-distortion probability (JEP) and the expected distortion reliability constraints. Under JEP, in the achievability part, we propose a type-based coding scheme, analyze the reliability guarantee for JEP and bound the leakage of the information source through compressed messages. In the converse part, by analyzing a guessing scheme of the eavesdropper, we prove the optimality of our achievability result. Under expected distortion, the achievability part is established similarly to the JEP counterpart. The converse proof proceeds by generalizing the corresponding results for the rate-distortion setting of SCS by Schieler and Cuff (T-IT 2014) to the successive refinement setting. Somewhat surprisingly, the normalized maximal leakage regions under both JEP and expected distortion constraints are identical under certain conditions, although JEP appears to be a stronger reliability constraint.
\end{abstract}


\section{Introduction}\label{sec:intro}
The Shannon cipher system (SCS)~\cite{shannon1949secrecy} is a classical model in information-theoretic secrecy, where a transmitter and a legitimate receiver are connected via a noiseless channel and share a secret key to achieve secure communication. The eavesdropper, named Eve, has access to the public channel as well as the source distribution and encryption schemes. To achieve perfect secrecy in SCS, which requires that the source sequence and the eavesdropped message are statistically independent, a necessary and sufficient condition is that the entropy of the secret key is no less than the entropy of the source sequence. One would desire the secret key to be uniformly distributed for security and thus the entropy of the secrecy key equals to the logarithm of its length. However, the shared secret key usually has a limited \emph{finite} length and is updated infrequently in practical communication systems, which is insufficient to ensure perfect secrecy.

To resolve the above problem, inspired by Shannon's seminal work~\cite{shannon1949secrecy}, several studies~\cite{yamamoto1997,merhav1999,schieler2014,schieler2016,Weinberger2017,issa2017guess} considered partial secrecy for SCS given a key rate under different security measures. Specifically, Yamamoto~\cite{yamamoto1997} adopted a distortion-based approach where the secrecy was measured by the minimum expected distortion at the eavesdropper. Merhav and Arikan~\cite{merhav1999} measured the secrecy of SCS by the number of guesses needed for the eavesdropper to successively reproduce the source sequence. Schieler and Cuff~\cite{schieler2016} proposed to consider the expected minimum distortion over an exponentially-sized list of estimates generated by the eavesdropper. Weinberger and Merhav~\cite{Weinberger2017} and Issa and Wagner~\cite{issa2017guess} measured the secrecy by the probability of successfully guessing the source sequence within a target distortion level.

Recently, Issa, Wagner and Kamath~\cite{issa2020leakage} introduced a new metric---\emph{maximal leakage}, into a threat model that captures several setups including SCS. In particular, the authors of \cite{issa2020leakage} derived the optimal asymptotic limit of the normalized maximal leakage for lossy compression of a discrete memoryless source (DMS) under both the excess-distortion probability and the expected distortion constraints. As an information measure, maximal leakage quantifies the maximal logarithmic gain in guessing any function of the original data from the public messages over random guessing. Furthermore, maximal leakage satisfies several axiomatic properties including data processing inequality, additivity property and independence property (cf.~\cite[Section I]{issa2020leakage}).

Maximal leakage has advantages in several aspects over other metrics. Compared with the mutual information measure, using maximal leakage can better characterize the severity of information leakage. As discussed in~\cite[Example 8]{issa2020leakage}, consider the alphabet $\calX=\{0,1\}^{8n}$ for an integer $n\in\bbN$ and let the random variable $X$ be distributed uniformly over $\calX$. Let $Y$ be the random variable that equals to $X$ if $X\;\mathrm{mod}\;8=0$ and equals to $1$ otherwise. Furthermore, let $Z$ be the first $n+1$ bits of $X$. Given the above setting, using the random variable $Y$, one can guess the random variable $X$ correctly with probability of at least $\frac{1}{8}$ while the probability of guessing $X$ correctly from $Z$ is only $2^{-7n+1}$. However, measuring the leakage via mutual information is somewhat contrary to intuition since $I(X;Y)\approx(n+0.169)\log2<I(X;Z) \approx(n+1)\log2$. It can be verified that the maximal leakage measure (cf.  Definition~\ref{def:maximal_leakage_DMS}) is consistent with the probability of correct guessing since $\rmL(X\to Y)=\log(2^{8n-3}+1)\geq\rmL(X\to Z)=(n+1)\log 2$. Note that expected distortion~\cite{yamamoto1997} and the expected number of guesses~\cite{merhav1999} could also predicate some insecure system as secure (cf.~\cite[Example 2]{issa2020leakage}). Finally, the threat model of maximal leakage has fewer assumptions about the eavesdropper while~\cite{Weinberger2017,issa2017guess} assume that the eavesdropper has access to the distortion measure and even the target distortion level shared by the encoder and the decoder. Due to the above advantages, maximal leakage has been adopted in various settings as the secrecy/privacy measure, e.g., membership privacy~\cite{Saeidian2021TIFS}, biometric template protection~\cite{Otroshi2023TIFS}, and information retrieval~\cite{Yakimenka2022JSAC}. For a comprehending review of various secrecy metrics, readers can refer to the surveys of Bloch et al.~\cite{bloch2021overview} and Hsu et al.~\cite{hsu2021survey}.

{All above works on SCS were based on the point-to-point source coding model while the characterization of the information leakage for multi-terminal models is missing. A representative multiterminal source coding problem is successive refinement~\cite{rimoldi1994,koshelev1981estimation,equitz1991successive}. This problem is an information-theoretic formulation of whether it is possible to decompose a lossy compression task with a target distortion level into multiple lossy compression tasks with decreasing distortion levels without loss of performance. Successive refinement has found diverse applications including clinical diagnosis using X-rays and image/video compression~\cite{rimoldi1994}. To evaluate the reliability of a code for successive refinement, there are two performance criteria: joint excess-distortion probability (JEP) and expected distortion. The JEP criterion quantifies the probability where either decoder fails to reconstruct the source sequence within the desired distortion level. The expected distortion criterion requires the expected distortion between the source sequence and the reproduced sequences of both decoders to be bounded by desired distortion levels. For DMS, Rimoldi~\cite{rimoldi1994} derived the rate-distortion region that asymptotically characterizes the optimal rate requirements of both encoders with vanishing JEP. The results of\cite{rimoldi1994} were subsequently refined by Kanlis and Narayan~\cite{kanlis1996error} who showed that the JEP vanishes exponentially as the blocklength $n$ increases for any rate pair strictly inside the rate-distortion region. Koshelev~\cite{koshelev1981estimation} and Equitz and Cover~\cite{equitz1991successive} studied the conditions for successive refinability, where optimal compression rates for both decoders can be simultaneously achieved as if the optimal codes are separately used for two point-to-point rate-distortion problems. Under JEP, Zhou, Tan and Motani~\cite{zhou2016second} refined Rimodi's results by deriving second-order asymptotics and moderate deviation asymptotics for DMS and Gaussian memoryless sources (GMS).  Bai, Wu and Zhou~\cite{wu2023TIT} further derived refined asymptotics of successive refinement for arbitrary memoryless sources using Gaussian codebooks under JEP. Tian et al.~\cite{tian2008SRBC} studied the Gaussian broadcast channel using a successive refinement code for GMS under the expected distortion constraint.

One might then wonder whether it is possible to generalize the SCS to the successive refinement model, i.e., the successive refinement model with an eavesdropper who has access to the messages from both encoders. Under this setting, one can study the trade-off between reliability, e.g., the coding performance, and secrecy, e.g., information leakage to the eavesdropper from the messages sent by two encoders.

%\vspace{-.3cm}
\subsection{Main Contributions}\label{sec:main_contribution}
In this paper, we answer the above question by studying the successive refinement setting of Shannon cipher system under maximal leakage for DMS under bounded distortion measures.  We adopt maximal leakage as the secrecy metric since it has advantages over other metrics in measuring the secrecy in SCS as mentioned in the 4th paragraph of Section~\ref{sec:intro}. To measure the reliability of a code, we consider two performance criteria: JEP and expected distortion. Under both criteria, we derive the asymptotic optimal normalized maximal leakage region under mild conditions.

Under JEP, we propose a type-based coding scheme and characterize the asymptotically achievable normalized maximal leakage region. By analyzing a guessing scheme of the eavesdropper, we prove the optimality of the achievable results under mild conditions. Both achievability and converse results are established by extending the point-to-point result~\cite[Theorem 8]{issa2020leakage} to the successive refinement setting. Our results reveal the fundamental trade-off between reliability and secrecy in the proposed model. When achievability and converse regions match, our coding scheme satisfies the successive refinability under maximal leakage if the source-distortion pair is successively refinable. In this case, there is no additional information leakage for successive refinement compared with rate-distortion in~\cite{issa2020leakage} if one aims to compress the source at the same distortion level.

Under expected distortion, inspired by~\cite[Theorem 9]{issa2020leakage}, we establish the achievable asymptotic normalized maximal leakage region by proposing a rate-distortion code. Using the fact that maximal leakage equals to the Sibson mutual information of order infinity for DMS~\cite[Theorem 1]{issa2020leakage}, we show that the above bound is tight under mild conditions. To do so, we generalize~\cite[Theorem 1, Corollary 5]{schieler2014}, where the the secrecy of rate-distortion with SCS is measured by equivocation, to the successive refinement setting. Furthermore, we show that for DMS satisfying certain conditions, the normalized maximal leakage regions under both expected distortion and JEP are identical, although the expected distortion constraint appears to be a looser criterion.

We next clarify our contributions beyond~\cite{issa2020leakage}. Note that the authors of~\cite{issa2020leakage} studied the Shannon cipher system, which corresponds to point-to-point lossy compression with a secrecy constraint. In contrast, we study the more general successive refinement setting with one more pair of encoder and decoder, which could model multi-user secret communication using keys. Due to the complication of the system model, the proof techniques for both the achievability and converse parts are different, especially in the analyses of the additional layer of encoder and decoder. Firstly, in the achievability part, we need to construct a coding scheme using the tailored type-covering lemma for the successive refinement problem and analyze the leakage of the source sequence from the encoded messages of both encoders. Secondly, in the converse analyses, under both JEP and expected distortion constraints, additional techniques are required to tackle two pairs of encoders and decoders. Specifically, under JEP, in order to analyze the eavesdropper's probability of correctly guessing the source sequence, we judiciously tailor the technique for SCS to the multiuser successive refinement setting, as detailed in Section~\ref{sec:Conv_DMS_proof_lemma_converse_DMS}. Under expected distortion, for the point-to-point SCS setting, the authors directly applied the existing result in \cite{schieler2014} and used the relationship between maximal leakage and mutual information. However, the corresponding result of \cite{schieler2014} for successive refinement is not available. As a result, we study the successive refinement of SCS with causal disclosure, establish the corresponding rate-equivocation region and finally derive an outer bound for the maximal leakage region, as detailed in Section~\ref{sec:Conv_DMS_expectation}.

%\vspace{-.3cm}
\subsection{Organization of the Rest of the Paper}
The rest of the paper is organized as follows. In Section~\ref{sec:problem fomulation}, we set up the notation and present the system model of successive refinement of SCS with necessary definitions. In Section~\ref{sec:main_results}, we present our main results and corresponding remarks. The achievability and converse proofs under JEP are presented in Sections~\ref{sec:Ach_TheoremDMS} and \ref{sec:Conv_TheoremDMS}, respectively. The proof of the results under expected distortion are provided in Section~\ref{sec:Conv_DMS_expectation}. Finally, in Section~\ref{sec:conclusion}, we conclude the paper and discuss future directions.

\section{Problem Formulation and Definitions}
\label{sec:problem fomulation}

\subsection*{Notation}
Random variables are in capital (e.g., $X$) and their realizations are in lower case (e.g., $x$). Random vectors of length $n$ and their particular realizations are denoted as $X^n:= (X_1, \ldots, X_n)$ and $x^n=(x_1,\ldots,x_n)$, respectively.  We use $\bbR$, $\bbR_+$, $\bbN$ to denote the set of real numbers, positive real numbers and integers respectively. We use calligraphic font (e.g., $\mathcal{X}$) to denote all other sets. For any two integers $(a,b)\in\bbN^2$, we use $[a:b]$ to denote the set of integers between $a$ and $b$, and we use $[a]$ to denote $[1:a]$. We use $\exp\{x\}$ to denote $2^x$ and use $\{x\}^+$ to denote $\max\{0,x\}$. All logarithms are base $2$. For any $p\in(0,1)$, we use  $H_b(p)$ to denote the binary entropy function $-p\log p-(1-p)\log(1-p)$. The set of all probability distributions on an alphabet $\calX$ is denoted by $\calP(\calX)$. For method of types, given a sequence $x^n$, we use $Q_{x^n}$ to denote its empirical distribution. The set of types formed from length-$n$ sequences taking values in $\calX^n$ is denoted as $\calQ_{\calX}^n$. Given $Q_X\in\calQ_{\calX}^n$, the set of all sequences of length-$n$ with type $Q_X$, also known as the type class, is denoted by $\calT_{Q_X}^n$.

\subsection{Problem Formulation}
As illustrated in Fig. \ref{fig:systemmodel}, we study the successive refinement setting of the Shannon cipher system~\cite{shannon1949secrecy}. Consider a memoryless source $X^n$ with distribution $P$ fully supported on the discrete alphabet $\calX$. For $i\in[2]$, the encoder $f_i$ and the decoder $\phi_i$ share a key $K_i^n$. The key $K_1^n$ is shared by both encoders $(f_1,f_2)$ and both decoders $(\phi_1,\phi_2)$ while $K_2^n$ is only shared by $f_2$ and $\phi_2$. Using $M_1$ and $K_1^n$, the decoder $\phi_1$ aims to reproduce the source sequence within distortion level $D_1$ and the reproduced source sequence $\hatX_1^n$ takes values in $\hatcalX_1^n$. With additional access to $M_2$ and $K_2^n$, the decoder $\phi_2$ aims to obtain a finer reproduction of the source sequence within distortion level $D_2<D_1$. The eavesdropper Eve aims to guess a random function of the source sequence $X^n$, denoted by $U$\footnote{Note that $P_{U|X}$ is unknown to the system designer, which can model various eavesdroppers. Such a setting ensures minimal assumptions about the eavesdropper.}, with knowledge of the compressed messages $(M_1,M_2)$, the source distribution and the encoding and decoding functions.

Let $(n,R_1,R_2,r_1,r_2)\in\bbN\times\bbR_+^4$ be arbitrary.
\begin{definition}\label{def:nR1R2_code}
An $(n,R_1,R_2,r_1,r_2)$-code for successive refinement Shannon cipher system consists of \footnote{Without loss of generality, we ignore the integer constraint for $nr_i$ and $nR_i$ with $i\in[2]$.}
\begin{itemize}
\item keys $K_i^n\in\calK_i^n=\{0,1\}^{nr_i}$ that are uniformly distributed over $\calK_i^n$ for each $i\in[2]$, where $r_i$ is the rate of the key $K_i^n$,
\item two encoders:
\begin{align}
f_1&:\calX^n\times\calK_1^n\to\calM_1:=[2^{nR_1}],\\*
f_2&:\calX^n\times\calK_2^n\to\calM_2:=[2^{nR_2}],
\end{align}
where $R_i$ is the rate of message $M_i$,
\item and two decoders:
\begin{align}
\phi_1&:\calM_1\times\calK_1^n\to\hat{\calX}_1,\\*
\phi_2&:\calM_1\times\calM_2\times\calK_1^n\times\calK_2^n \to\hat{\calX}_2.
\end{align}
\end{itemize}
\end{definition}
For ease of notation, given each $i\in[2]$, we use $\vec{R_i}$ to denote the rate pair $(R_i,r_i)$. This way, an $(n,R_1,R_2,r_1,r_2)$-code is equivalent to an $(n,\vec{R_1},\vec{R_2})$-code.

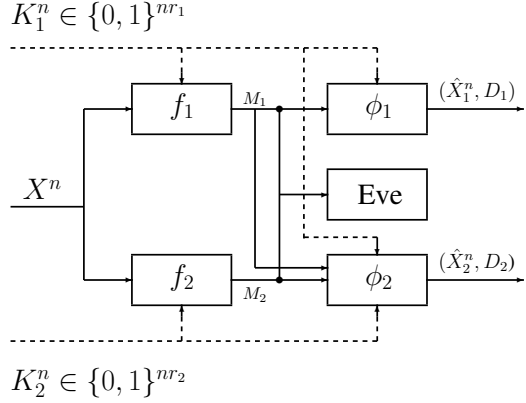
\begin{figure}[t]
\centering
\setlength{\unitlength}{0.5cm}
\scalebox{0.65}{
\begin{picture}(22,18)
\linethickness{1pt}
\put(1.5,8.8){\makebox{\LARGE$X^n$}}
%Keys
\multiput(1,3)(0.4,0){38}{\line(1,0){0.2}}
\put(1,1){\makebox{\LARGE$K_2^n\in\{0,1\}^{nr_2}$}}
\multiput(8,3)(0,0.4){3}{\line(0,1){0.2}}
\put(8,4){\vector(0,1){0.5}}
\multiput(16,3)(0,0.4){3}{\line(0,1){0.2}}
\put(16,4){\vector(0,1){0.5}}
\multiput(1,15)(0.4,0){38}{\line(1,0){0.2}}
\put(1,16){\makebox{\LARGE$K_1^n\in\{0,1\}^{nr_1}$}}
\multiput(8,15)(0,-0.4){3}{\line(0,-1){0.2}}
\put(8,14){\vector(0,-1){0.5}}
\multiput(16,15)(0,-0.4){3}{\line(0,-1){0.2}}
\put(16,14){\vector(0,-1){0.5}}
\multiput(13,15)(0,-0.4){20}{\line(0,-1){0.2}}
\multiput(13,7.25)(0.4,0){8}{\line(1,0){0.2}}
\put(16,7.25){\vector(0,-1){0.75}}
%put two encoders
\put(6,4.5){\framebox(4,2)}
\put(6,11.5){\framebox(4,2)}
%put encoder names
\put(7.5,5.3){\makebox{\LARGE$f_2$}}
\put(7.5,12.3){\makebox{\LARGE$f_1$}}
%line the source to encoders
\put(1,8.5){\line(1,0){3}}
\put(4,12.5){\vector(1,0){2}}
\put(4,5.5){\line(0,1){7}}
\put(4,5.5){\vector(1,0){2}}
%put two decoders and Eve
\put(14,4.5){\framebox(4,2)}
\put(14,8){\framebox(4,2)}
\put(14,11.5){\framebox(4,2)}
%put decoder names
\put(15.5,12.3){\makebox{\LARGE$\phi_1$}}
\put(15.5,5.3){\makebox{\LARGE$\phi_2$}}
%put Eve
\put(15.2,8.6){\makebox{\LARGE Eve}}
%put arrows to decoders
\put(10,5.5){\vector(1,0){4}}
\put(11,4.9){\makebox(0,0){$M_2$}}
\put(10,12.5){\vector(1,0){4}}
\put(11,12.9){\makebox(0,0){$M_1$}}%\in\{0,1\}^{nR_1}
\put(11,12.5){\line(0,-1){6.5}}
\put(11,6){\vector(1,0){3}}
%line Eve
\put(12,12.5){\line(0,-1){7}}
\put(12,12.5){\circle*{0.3})}
\put(12,5.5){\circle*{0.3})}
\put(12,9){\vector(1,0){2}}
%put decoded symbols
\put(18,5.5){\vector(1,0){4}}
\put(18.5,6){\makebox{\large$(\hatX_2^n,D_2$)}}
\put(18,12.5){\vector(1,0){4}}
\put(18.5,13){\makebox{\large$(\hatX_1^n,D_1)$}}
\end{picture}}
\caption{Successive refinement of Shannon cipher system with an eavesdropper.}
\label{fig:systemmodel}
\end{figure}

\subsection{Definitions of Reliability and Secrecy Criteria}
For $i\in[2]$, define two bounded distortion measures: $d_i:\calX\times\hat{\calX}_i\to[0,\infty)$ such that for each $x\in\calX$, there exists $\hatx_i\in\hat{\calX}_i$ satisfying $d_i(x,\hatx_i)=0$. Furthermore, the distortion between $x^n$ and $\hatx_i^n$ is defined as $d_i(x^n,\hatx_i^n):= \frac{1}{n}\sum_{j=1}^{n}d_i(x_j,\hatx_{i,j})$. For any $(D_1,D_2)\in\bbR_+^2$, the joint-excess-distortion probability (JEP) is defined as follows:
\begin{align}
\rmP_{\rm{e}}^n(D_1,D_2):=\Pr\{d_1(X^n,\hatX_1^n)>D_1\;\mathrm{or}\;d_2(X^n,\hatX_2^n)>D_2 \}. \label{def:JEP}
\end{align}

To measure the information leakage of the source from compressed messages, we use use the following definition of maximal leakage for DMS in~\cite[Theorem 1]{issa2020leakage}:
\begin{definition}\label{def:maximal_leakage_DMS}
For any distribution $P_{XY}$ defined on the finite alphabet $\calX\times\calY$, maximal leakage from $X$ to $Y$ is defined as
\begin{align}
\rmL(X\to Y):=\log\sum\limits_{y\in\calY}\max\limits_{\substack{x\in\calX:\\P(x)>0}} P_{Y|X}(y|x).
\end{align}

\end{definition}

Maximal leakage has advantages over other secrecy metric, e.g., expected distortion and mutual information, when the threat results from the adversary who tries to guess sensitive information based on observed messages. For example, as discussed in~\cite[Examples 7,8]{issa2020leakage}, using maximal leakage is less likely to underestimate the risk of such threats.

Let $\alpha\in\bbR_+$ be arbitrary. The fundamental limit for successive refinement SCS under the JEP constraint is the normalized maximal leakage region, which is defined as follows.
\begin{definition}\label{def:leakage_pair}
A pair $(L_1,L_2)$ is said to be $(D_1,D_2,\vec{R_1},\vec{R_2},\alpha)$-achievable under the JEP constraint if there exists a sequence of $(n,\vec{R_1},\vec{R_2})$-codes such that
\begin{align}
\limsup_{n\to\infty}\frac{1}{n}\rmL(X^n\to M_1)&\leq L_1, \label{def:ach_L1}\\*
\limsup_{n\to\infty}\frac{1}{n}\rmL(X^n\to M_1M_2)&\leq L_2, \label{def:ach_L2}
\end{align}
and
\begin{align}
\rmP_{\rm{e}}^n(D_1,D_2)\leq2^{-n\alpha}. \label{def:ach_JEP}
\end{align}
The closure of the set of all $(D_1,D_2,\vec{R_1},\vec{R_2},\alpha)$-achievable normalized maximal leakage pairs is called $(D_1,D_2,\vec{R_1},\vec{R_2},\alpha)$-achievable normalized maximal leakage region and denoted as $\calL(D_1,D_2,\vec{R_1},\vec{R_2},\alpha|P)$.
\end{definition}
Definition \ref{def:leakage_pair} defines the achievable maximal leakage region subject to a JEP constraint. Note that the boundary of the region $\calL(D_1,D_2,\vec{R_1},\vec{R_2},\alpha|P)$ are determined by the asymptotic limits of $\frac{1}{n}\rmL(X^n\to M_1)$ and $\frac{1}{n}\rmL(X^n\to M_1M_2)$ for a sequence $(n,\vec{R_1},\vec{R_2})$-codes.

Another widely adopted reliability criterion for lossy source coding is expected distortion~\cite{shannon1959coding,Gray1998quantization}. Accordingly, the fundamental limit under expected distortion is defined as follows.
\begin{definition}\label{def:leakage_region_expected}
A pair $(L_1,L_2)$ is said to be $(D_1,D_2,\vec{R_1},\vec{R_2})$-achievable under expected distortion if there exists a sequence of $(n,\vec{R_1},\vec{R_2})$-codes such that
\begin{align}
\limsup_{n\to\infty}\frac{1}{n}\rmL(X^n\to M_1)&\leq L_1,\\
\limsup_{n\to\infty}\frac{1}{n}\rmL(X^n\to M_1M_2)&\leq L_2,
\end{align}
and
\begin{align}
\mathbf{E}\big[d_1(X^n,\hatX_1^n)\big]\leq D_1,\label{avg:const1}\\
\mathbf{E}\big[d_2(X^n,\hatX_2^n)\big]\leq D_2\label{avg:const2}.
\end{align}
The closure of the set of all $(D_1,D_2,\vec{R_1},\vec{R_2})$-achievable normalized maximal leakage pairs is called the $(D_1,D_2,\vec{R_1},\vec{R_2})$-achievable normalized maximal leakage region and denoted as $\calL_{\rm{exp}}(D_1,D_2,\vec{R_1},\vec{R_2}|P)$.
\end{definition}

One might wonder why we do not consider the leakage from the message $M_2$, i.e., $\rmL(X^n\to M_2)$. It appears that there is no difference between $M_1$ and $M_2$ from the point of view of the adversary. However, it follows from Definition~\ref{def:nR1R2_code} that one could not decode the source sequence correctly from $M_2$ without $M_1$ in the successive refinement setting. This indicates that no meaningful reliable performance analysis can be obtained by simply observing $M_2$. Furthermore, the leakage $\mathrm{L}(X^n\to M_1M_2)$ from $(M_1,M_2)$ is naturally an upper bound for the leakage $\mathrm{L}(X\to M_2)$ only from $M_2$ since the adversary has access to more information in the former case.

\section{Main Results}\label{sec:main_results}
\subsection{JEP Reliability Criterion}
Let $R(Q,D_1)$ be the rate-distortion function for the source distribution $Q$ and $R(Q,R_1,D_1,D_2)$ be the minimum sum rate of encoders $(f_1,f_2)$ subject to the rate constraint $R_1$ for encoder $f_1$  i.e.,
\begin{align}
R(Q,D_1)&:=\inf\limits_{Q_{\hatX_1|X}:\mathbb{E}[d_1(X,\hatX_1)]\leq D_1}I(X;\hatX_1),\label{def:R(QD1)}\\
R(Q,R_1,D_1,D_2)&:=\inf\limits_{\substack{Q_{\hatX_1\hatX_2|X}: \mathbb{E}[d_1(X,\hatX_1)]\leq D_1 \\\mathbb{E}[d_2(X,\hatX_2)]\leq D_2, I(X,\hatX_1)\leq R_1}} I(X;\hatX_1,\hatX_2). \label{def:R(QR1D1D2)}
\end{align}
%The rate constraints in Eq. \eqref{def:DMS_require_R1} and \eqref{def:DMS_require_R1R2} are necessary to guarantee that the JEP satisfies $\rmP_{\rm{e}}^n(D_1,D_2)\leq2^{-n\alpha}$, which will be proved in Section \ref{sec:DMS_ach_codingscheme}.

Furthermore, define the following exponent functions
\begin{align}
\Lambda_1(P,\vec{R_1},D_1,\alpha)&:=\max\limits_{Q:D(Q||P)\leq\alpha}\{R(Q,D_1)-r_1\}^+, \label{def:Lambda1}\\
\Lambda_2(P,\vec{R_1},\vec{R_2},D_1,D_2,\alpha) &:=\max\limits_{Q:D(Q||P)\leq\alpha}\big\{\{R(Q,D_1)-r_1\}^+ +\{R(Q,R_1,D_1,D_2)-R(Q,D_1)-r_2\}^+\big\},\label{def:Lambda2}\\ \Lambda_2^{\mathrm{out}}(P,\vec{R_1},\vec{R_2},D_1,D_2,\alpha) &:=\max\limits_{Q:D(Q||P)\leq\alpha}\{R(Q,R_1,D_1,D_2)-r_1-r_2\}^+.\label{def:Lambda2out}
\end{align}

Finally, for DMS with distribution $P$, given any $\alpha>0$, define the following regions
\begin{align}
\calL^{\mathrm{in}}(D_1,D_2,\vec{R_1},\vec{R_2},\alpha|P) &:=\Big\{(L_1,L_2):L_1\geq \Lambda_1(P,\vec{R_1},D_1,\alpha), L_2\geq \Lambda_2(P,\vec{R_1},\vec{R_2},D_1,D_2,\alpha)\Big\},\label{def:caL:in}\\
\calL^{\mathrm{out}}(D_1,D_2,\vec{R_1},\vec{R_2},\alpha|P) &:=\Big\{(L_1,L_2):L_1\geq \Lambda_1(P,\vec{R_1},D_1,\alpha), L_2\geq \Lambda_2^{\mathrm{out}}(P,\vec{R_1},\vec{R_2},D_1,D_2,\alpha)\Big\}\label{def:caL:out}.
\end{align}

\begin{theorem}\label{theo:DMS} %Consider the rate pair satisfying....
Consider the rate pair $(R_1,R_2)$ such that
\begin{align}
R_1&>\max\limits_{Q:D(Q||P)\leq\alpha}R(Q,D_1),\label{def:DMS_require_R1}\\*
R_1+R_2&>\max\limits_{Q:D(Q||P)\leq\alpha}R(Q,R_1,D_1,D_2).\label{def:DMS_require_R1R2}
\end{align}
The $(D_1,D_2,\vec{R_1},\vec{R_2},\alpha)$-achievable  maximal leakage region satisfies
\begin{align}
\calL^{\mathrm{in}}(D_1,D_2,\vec{R_1},\vec{R_2},\alpha|P) &\subseteq\calL(D_1,D_2,\vec{R_1},\vec{R_2},\alpha|P)\\
&\subseteq\calL^{\mathrm{out}}(D_1,D_2,\vec{R_1},\vec{R_2},\alpha|P).
\end{align}
\end{theorem}
The achievability (inner bound) and the converse (outer bound) proofs of Theorem \ref{theo:DMS} are provided in Sections~\ref{sec:Ach_TheoremDMS} and \ref{sec:Conv_TheoremDMS}, respectively.  We make the following remarks.

\begin{remark}
To prove the achievability part of Theorem \ref{theo:DMS}, we propose a type-based coding scheme using bitwise encryption, upper bound the normalized maximal leakage pairs by generalizing~\cite[Section IV-E]{issa2020leakage} and show that JEP decays exponentially fast using the method of types~\cite{csiszar1998mt} and the type covering lemma for successive refinement~\cite[Lemma 1]{kanlis1996error},~\cite[Lemma 8]{no2016}. To prove the converse part of Theorem~\ref{theo:DMS}, we derive a lower bound for the normalized maximal leakage between the source sequence $X^n$ and the messages $M_1$ and $M_2$. Specifically, inspired by~\cite[Section IV-E]{issa2020leakage}, we propose a guessing scheme for Eve and generalize~\cite[Lemma 5]{issa2017guess} to the successive refinement setting and bound the probability of correctly guessing the source sequence by Eve under the JEP constraint.
\end{remark}

\begin{remark}
The achievability part of Theorem~\ref{theo:DMS} generalizes the achievability part of~\cite[Theorem 8]{issa2020leakage} to the successive refinement setting and reveals the fundamental tradeoff between reliability and secrecy. For ease of notation, given $(P,\vec{R_1},\vec{R_2},D_1,D_2)$, we use $\Lambda_1(\alpha)$ and $\Lambda_2(\alpha)$ to denote $\Lambda_1(P,\vec{R_1},D_1,\alpha)$ and $\Lambda_2(P,\vec{R_1},\vec{R_2},D_1,D_2,\alpha)$, respectively. Note that $\Lambda_i(\alpha)$ is a non-decreasing function of $\alpha$ for each $i\in[2]$. Thus, generally speaking, a looser reliability constraint with a smaller JEP exponent $\alpha$ leads to better secrecy guarantee with less information leakage. Furthermore, there exists a floor effect, where the secrecy guarantee remains unchanged if the reliability constraint $\alpha$ is above a certain threshold. To illustrate, for each $i\in[2]$, let $\alpha_i^*$ be the minimum $\alpha_i\in\bbR_+$ such that for all $\alpha\geq\alpha_i$,
\begin{align}
\Lambda_i(\alpha)=\Lambda_i(\alpha_i).
\end{align}
If $\alpha\geq\alpha_i^*$, $\Lambda_i(\alpha)$ remains unchanged. This implies that, when the reliability constraint $\alpha$ is above a certain threshold, regardless of the reliability constraint, the normalized maximal leakage remains the same. We provide a numerical example to further illustrate this point in Remark~\ref{remark:numerical_example}.
\end{remark}

\begin{remark}
The converse part of Theorem~\ref{theo:DMS} generalizes the converse part of~\cite[Theorem 8]{issa2020leakage} to the successive refinement setting. The remark of the inner bound of Theorem~\ref{theo:DMS} is also valid for the outer bound, with a slight change where $\Lambda_2$ is replaced by $\Lambda_2^{\mathrm{out}}$.
\end{remark}

As shown in the following corollary, our achievability and converse bounds match under mild conditions.
\begin{corollary}\label{theo:coro_JEP_match}
Consider key rate pairs $(r_1,r_2)$ such that for all $Q$ satisfying $D(Q||P)\leq\alpha$,
\begin{align}
R(Q,D_1)&\geq r_1, \label{DMScoro:condition_r1}\\
R(Q,R_1,D_1,D_2)-R(Q,D_1)&\geq r_2 \label{DMScoro:condition_r2}.
\end{align}
It follows that
\begin{align}
\calL^{\mathrm{in}}(D_1,D_2,\vec{R_1},\vec{R_2},\alpha|P) =\calL^{\mathrm{out}}(D_1,D_2,\vec{R_1},\vec{R_2},\alpha|P).
\end{align}
\end{corollary}

\begin{remark}\label{remark:condition_partial_secrecy}
The conditions in Eq.~\eqref{DMScoro:condition_r1} and \eqref{DMScoro:condition_r2} are mild since the key rates $r_1$ and $r_2$ are usually limited. It follows from the codebook design in Section~\ref{sec:DMS_ach_codebook} of the achievability proof that given a type $Q$ of the source sequence $X^n$, the number of codewords used by the first encoder is roughly upper bounded by $2^{R(Q,D_1)}$. Thus,  Eq.~\eqref{DMScoro:condition_r1} implies that the key rate $r_1$ will not exceed the rate $R_1$ of the first encoder under the JEP constraint. The above statement also holds for $r_2$ similarly since the number of codewords used by the second encoder is roughly upper bounded by $2^{R(Q,R_1,D_1,D_2)-R(Q,D_1)}$. In other words, Eq.~\eqref{DMScoro:condition_r1} and \eqref{DMScoro:condition_r2} correspond to partial secrecy. As a sanity check, consider a Bernoulli source with distribution $P=\mathrm{Bern}(0.4)$ under Hamming distortion measures. When $D_1=0.2$, $D_2=0.15$ and $\alpha=0.03$. the conditions on key rates are $r_1\leq 0.162$ and $r_2\leq 0.112$.
\end{remark}

\begin{remark}\label{remark:successive_refinability}
Under the conditions of Corollary \ref{theo:coro_JEP_match}, it follows that for each $i\in[2]$, $\lim_{\alpha\to\infty}\Lambda_i(\alpha)=\Lambda_i(\alpha_i^*)$ and $\Lambda_2(\alpha_2^*)\geq \Lambda_1(\alpha_1^*)$. This way, we can discuss the \emph{successive refinability} of the maximal leakage pair. For successive refinement, a source-distortion triplet is said to be successively refinable if one can simultaneously achieve the minimal compression rates for both decoders as if the compression is done  separately~\cite{equitz1991successive,koshelev1981estimation}, i.e., $R(P,R(P,D_1),D_1,D_2)=R(P,D_2)$ for all $P\in\calP(\calX)$. If the source-distortion measure triplet is successively refinable, it follows that $\Lambda_2^{\mathrm{out}}(P,\vec{R_1},\vec{R_2},D_1,D_2,\alpha)=\Lambda_1(P,\vec{R_1}+\vec{R_2},D_2,\alpha)$. Note that $\Lambda_1(P,\vec{R_1}+\vec{R_2},D_2,\alpha)$ is the minimal normalized maximal leakage when one aims to achieve the distortion level $D_2$ in point-to-point SCS~\cite[Theorem 8]{issa2020leakage}. Thus, the above result implies that the proposed scheme in the successive refinement setting of SCS has the same secrecy and reliable performance as the point-to-point SCS setting under the same distortion level and the same excess-distortion probability constraint. In other words, successive refinability for the pure source coding problem extends to the SCS setting under maximal leakage.
\end{remark}

\begin{remark}\label{remark:numerical_example}
Consider a DMS with distribution $P=\mathrm{Bern}(p)$ under Hamming distortion measures. Such a source distortion triple is successively refinable~\cite{koshelev1981estimation,equitz1991successive}, i.e., $R(P,R(P,D_1),D_1,D_2)=R(P,D_2)$. It follows that $\alpha^*:=\alpha_1^*=\alpha_2^*$. This is because for any $D\in\bbR_+$, the optimization problem
\begin{align}
\max_{Q:D(Q\|P)\leq \alpha} R(Q,D)=\max_{Q:D_b(q\|p)\leq \alpha} H_b(q)-H_b(D) \label{equi:binary_RQD}
\end{align}
has the same maximizer, where $H_b(q):=-q\log q-(1-q)\log (1-q)$ denotes the binary entropy and $D_b(q\|p):=q\log\frac{q}{p}+(1-q)\log\frac{1-q}{1-p}$ denotes the binary relative entropy. Thus, the achievable maximal leakage region satisfies
\begin{align}
\Lambda_1(\alpha)&=\max\limits_{q:D_b(q\|p)\leq\alpha}\{H_b(q)-H_b(D_1)-r_1\}^+,\\
\Lambda_2(\alpha)&=\max\limits_{q:D_b(q\|p)\leq\alpha}\{H_b(q)-H_b(D_2)-r_1-r_2\}^+.
\end{align}
Note that $H_\rmb(q)$ achieves the maximum value of $1$ when $q=0.5$. Thus, when $q=0.5$ is feasible in both optimization problems above, the values of $\Lambda_1(\alpha)$ and $\Lambda_2(\alpha)$ remain unchanged. In turn, this requires that $\alpha\geq \alpha^*=D_b(0.5\|p)$. In Fig. \ref{fig:L1L2_plot}, we numerically illustrate $(\Lambda_1(\alpha),\Lambda_2(\alpha))$ when $p=0.3$. As observed from Fig. \ref{fig:L1L2_plot}, $\Lambda_i(\alpha)$ increases in $\alpha$ when $\alpha<\alpha^*$ and converges when $\alpha\geq \alpha^*$. The converged values of $\Lambda_i(\alpha)$ satisfies that $\Lambda_1(\alpha^*)=H_b(0.5)-H_b(D_1)-r_1$ and $\Lambda_2(\alpha^*)=H_b(0.5)-H_b(D_2)-r_1-r_2$. In Fig.~\ref{fig:L1_surf_tangent}, we plot $\Lambda_1(\alpha)$ for various values of the reliability constraint $\alpha$ and the key rate $r_1$. In Fig. \ref{fig:L1_surf_tangent}, we also plot the slice of the 3-D surface of $\Lambda_1(\cdot)$ for a fixed key rate $r_1=0.06$, which corresponds to the blue curve in Fig.~\ref{fig:L1L2_plot}.
\end{remark}

\begin{figure}
\centering
\includegraphics[width=.5\columnwidth]{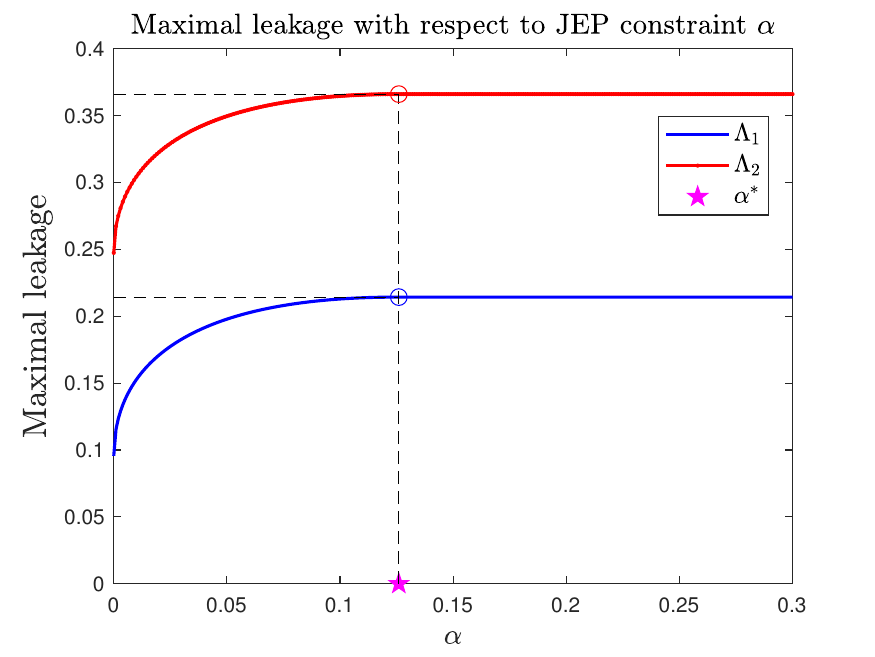}
\caption{Illustration of the boundaries of maximal leakage region $(\Lambda_1,\Lambda_2)$ with respect to the JEP constraint $\alpha$ for $P=\mathrm{Bern}(0.3)$, $D_1=0.2$, $D_2=0.1$, $r_1=0.06$ and $r_2=0.1$.}
\label{fig:L1L2_plot}
\end{figure}

\begin{figure}[tb]
\centering
\includegraphics[width=.5\columnwidth]{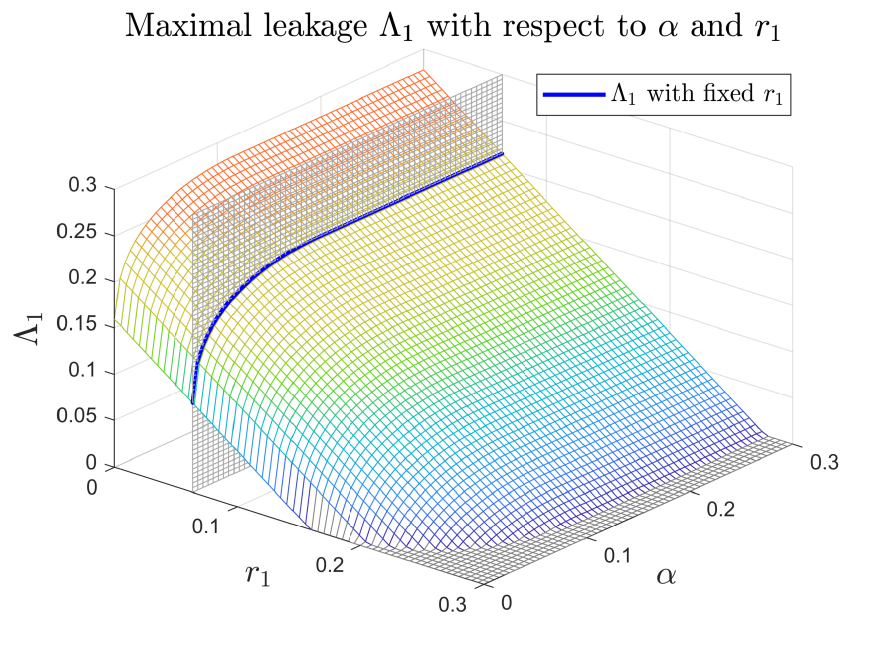}
\caption{Illustration of $\Lambda_1$ with respect to $\alpha$ and $r_1$ for $P=\mathrm{Bern}(0.3)$ and $D_1=0.2$. The slice of $\Lambda_1$ is determined by a fixed key rate $r_1=0.06$.}
\label{fig:L1_surf_tangent}
\end{figure}

%\blue{\begin{remark}
%For a successively refinable source distortion triple, whether $\alpha_1^*=\alpha_2^*$ holds mainly depends on whether the rate-distortion function can separate the parameters of source distribution and distortion during optimization. For example, for an equiprobable DMS $P=\frac{1}{|\calA_P|}$, the rate-distortion function is given by\red{reference}
%\begin{align}
%R(P,D)=\log|\calA_P|-H_b(D)-D\log(|\calA_P|-1),
%\end{align}
%for $D\in[0,1-\frac{1}{|\calA_P|}]$. The maximizer is given by
%\begin{align}
%\argmax_{Q:D(Q||P)\leq\alpha}\log|\calA_Q|-D\log(|\calA_Q|-1),
%\end{align}
%which isn't affected by $D$. \red{However, for a erased binary source}
%\end{remark}}

\subsection{Expected Distortion Reliability Criterion}
In this section, we consider the expected distortion reliability criterion. Define the following exponent functions,
\begin{align}
\Omega_1(P,\vec{R_1},D_1)&:=\{R(P,D_1)-r_1\}^+,\\
\Omega_2(P,\vec{R_1},\vec{R_2},D_1,D_2)&:=\{R(P,D_1)-r_1\}^+ +\{R(P,R_1,D_1,D_2)-R(P,D_1)-r_2\}^+,\\
\Omega_2^{\rm{out}}(P,\vec{R_1},\vec{R_2},D_1,D_2)&:=\{R(P,R_1,D_1,D_2)-r_1-r_2\}^+.
\end{align}

Finally, for DMS with distribution $P$, define the following regions
\begin{align}
\calL_{\rm{exp}}^{\rm{in}}(D_1,D_2,\vec{R_1},\vec{R_2}|P) &:=\Big\{(L_1,L_2):L_1\geq \Omega_1(P,\vec{R_1},D_1), L_2\geq \Omega_2(P,\vec{R_1},\vec{R_2},D_1,D_2)\Big\},\label{calL;avg:in}\\
\calL_{\rm{exp}}^{\mathrm{out}}(D_1,D_2,\vec{R_1},\vec{R_2}|P) &:=\Big\{(L_1,L_2):L_1\geq \Omega_1(P,\vec{R_1},D_1), L_2\geq \Omega_2^{\mathrm{out}}(P,\vec{R_1},\vec{R_2},D_1,D_2)\Big\}\label{calL;avg:out}.
\end{align}

Our achievability result states as follows.
\begin{theorem}\label{theo:DMS_expectation}
Consider the rate pair $(R_1,R_2)$ such that
\begin{align}
R_1&>R(P,D_1), \label{DMSexpect:req_R1}\\*
R_1+R_2&>R(P,R_1,D_1,D_2).\label{DMSexpect:req_R1R2}
\end{align}
The $(D_1,D_2,\vec{R_1},\vec{R_2})$-achievable maximal leakage region satisfies
\begin{align}
\calL_{\rm{exp}}^{\rm{in}}(D_1,D_2,\vec{R_1},\vec{R_2}|P) &\subseteq\calL_{\rm{exp}}(D_1,D_2,\vec{R_1},\vec{R_2}|P)\\
&\subseteq\calL_{\rm{exp}}^{\mathrm{out}}(D_1,D_2,\vec{R_1},\vec{R_2}|P).
\end{align}
\end{theorem}

The proof of Theorem~\ref{theo:DMS_expectation} is provided in Section~\ref{sec:Conv_DMS_expectation}. We make a few remarks.

\begin{remark}
We can compare the achievable normalized maximal leakage regions under JEP and expected distortion reliability constraints. The achievability proof of Theorem~\ref{theo:DMS_expectation} is implied by Theorem~\ref{theo:DMS} with JEP reliability constraint. Thus, expected distortion is a weaker reliability constraint than JEP. One might wonder whether such a weaker constraint leads to a better secrecy guarantee, i.e., for fixed $(D_1,D_2,\vec{R_1},\vec{R_2})$, whether the normalized maximal leakage region under expected distortion is a strict subset of the corresponding region under JEP. \emph{Counter-intuitively}, our result provides a negative answer. Consider a Bernoulli source distribution with parameter $p=0.5$, i.e., $P=\mathrm{Bern}(0.5)$. It follows from Eq. \eqref{equi:binary_RQD} that $\alpha_1^*=\alpha_2^*=0$. It follows that $\Omega_1(P,\vec{R_1},D_1)=\Lambda_1(P,\vec{R_1},D_1,\alpha)$ and $\Omega_2(P,\vec{R_1},\vec{R_1},D_1,D_2)=\Lambda_2(P,\vec{R_1},\vec{R_1},D_1,D_2,\alpha)$ for any $\alpha>0$. Therefore, the achievable maximal leakage region under both JEP and expected distortion constraints are identical in this case.
\end{remark}

\begin{remark}
To prove the converse part, inspired by~\cite[Theorem 9]{issa2020leakage}, we first lower bound normalized maximal leakage via normalized mutual information and then further lower bound mutual information by deriving an upper bound on equivocation. Our main contribution is to establish the converse results for the successive refinement setting of SCS under the equivocation secrecy measure. To do so, we generalize the converse proof of the rate-distortion equivocation region~\cite[Theorem 1 and Corollary 5, part 1]{schieler2014} to the successive refinement setting, which is detailed in Section~\ref{sec:Conv_DMS_proof_lemma_converse_DMS}.
\end{remark}

Analogous to Corollary \ref{theo:coro_JEP_match}, under the following conditions, our bounds under expected distortion match.
\begin{corollary}\label{theo:coro_expected_match}
If the key rate pairs $(r_1,r_2)$ satisfy
\begin{align}
R(P,D_1)&\geq r_1,\label{DMSexpcoro:condition_r1}\\
R(P,R_1,D_1,D_2)-R(P,D_1)&\geq r_2, \label{DMSexpcoro:condition_r2}
\end{align}
it follows that
\begin{align}
\calL_{\rm{exp}}^{\mathrm{in}}(D_1,D_2,\vec{R_1},\vec{R_2}|P) =\calL_{\rm{exp}}^{\mathrm{out}}(D_1,D_2,\vec{R_1},\vec{R_2}|P).
\end{align}
\end{corollary}

\begin{remark}
The conditions in Eq. \eqref{DMSexpcoro:condition_r1} and \eqref{DMSexpcoro:condition_r2} are mild and correspond to partial secrecy. The maximal leakage is also successively refinable under expected distortion for successively refinable source-distortion triplets.
\end{remark}

\section{Achievability Proof of Theorem \ref{theo:DMS}}
\label{sec:Ach_TheoremDMS}
To establish the achievability part of Theorem \ref{theo:DMS}, we need to design a coding scheme satisfying the JEP constraint and characterize the normalized maximal leakage region of the proposed scheme. Firstly, we specify the codebook design. Subsequently, we explain the encoding and decoding scheme and prove the proposed coding scheme satisfies the JEP constraint. Finally, we analyze normalized maximal leakage between the source sequence and compressed messages for the coding scheme.

\subsection{Type Covering Lemma}\label{sec:DMS_ach_typecoveringlemma}
To state our encoding scheme, we recall the type covering lemma for successive refinement source coding in~\cite[Lemma 1]{kanlis1996error},~\cite[Lemma 8]{no2016},~\cite[Lemma 16]{zhou2016second}.
\begin{lemma}\label{theo:lemma_SR_DMS}
Define two constants:
\begin{align}
c_1&=4|\calX|\cdot|\hat{\calX}_1|+9,\label{def:c1}\\
c_2&=6|\calX|\cdot|\hat{\calX}_1|\cdot|\hat{\calX}_2|+2|\calX|\cdot|\hat{\calX}_1|+17. \label{def:c2}
\end{align}
Given type $Q_{X}\in\calQ_\calX^n$, for all $\tilde{R}_1\geq R(Q_X,D_1)$, the following holds:
\begin{itemize}
\item There exist a set $\calB_Y(Q_X)\subset\hat\calX_1^n$ such that
\begin{align}
\frac{1}{n}\log|\calB_Y(Q_X)|\leq \tilde{R}_1+c_1\frac{\log n}{n}
\end{align}
and $\calB_Y(Q_X)$ $D_1$-covers $\calT_{Q_X}^n$, i.e.,
\begin{align}
\calT_{Q_X}^n\subset\bigcup_{y^n\in\calB_Y(Q_X)}\calN_1(y^n,D_1),
\end{align}
where
\begin{align}
\calN_1(y^n,D_1):=\big\{x^n:d_1(x^n,y^n)\leq D_1\big\}.
\end{align}
\item For each $x^n\in\calT_{Q_X}^n$ and each $y^n\in\calB_Y(Q_X)$, there exists a set $\calB_Z(y^n)\subset\hat\calX_2^n$ such that
\begin{align}
\frac{1}{n}\log\left(\sum_{y^n\in\calB_Y(Q_X)}|\calB_Z(y^n)|\right) \leq R(Q_X,R_1,D_1,D_2)+c_2\frac{\log n}{n}
\end{align}
and $\calB_Z(y^n)$ $D_2$-covers $\calN_1(y^n,D_1)$, i.e.,
\begin{align}
\calN_1(y^n,D_1)\subset\bigcup_{z^n\in\calB_Z(y^n)}\calN_2(z^n,D_2),
\end{align}
where
\begin{align}
\calN_2(z^n,D_2):=\big\{x^n:d_2(x^n,z^n)\leq D_2\big\}.
\end{align}
\end{itemize}
\end{lemma}

\subsection{Codebook Design}\label{sec:DMS_ach_codebook}
Let $\epsilon_1=c_1\frac{\log n}{n}$, $\epsilon_2=c_2\frac{\log n}{n}$ and let $n$ be large enough. For each type $Q_X\in\calQ_{\calX}^n$, we construct a successive refinement code such that
\begin{itemize}
\item  every sequence $x^n\in\calT_{Q_X}^n$ is $D_1$-covered by a codebook $\calB_Y(Q_X)$ and $|\calB_Y(Q_X)|\leq2^{n(R(Q_X,D_1)+\epsilon_1)}$,
\item given $y^n\in\calB_Y(Q_X)$, $x^n$ is $D_2$-covered by a codebook $\calB_Z(y^n)$ and $\sum_{y^n\in\calB_Y(Q_X)}|\calB_Z(y^n)|\leq 2^{n(R(Q_X,R_1,D_1,D_2)+\epsilon_2)}$.
\end{itemize}
Such construction is guaranteed by Lemma \ref{theo:lemma_SR_DMS}. Given a source sequence $x^n$, we use $y^n$ to denote the codeword output by the first encoder and use $z^n$ to denote the codeword output by the second encoder.

We next divide the codebook $\calB_Y(Q_X)$ into $\lceil|\calB_Y(Q_X)/2^{nr_1}|\rceil$ bins, each of size $2^{nr_1}$, except for possibly the last one. Then, we use $\calB_Y(Q_X,i,\cdot)$ to denote the $i$th partition of the codebook $\calB_Y(Q_X)$ and $\calB_Y(Q_X,i,j)$ to denote the $j$-th codeword in the $i$-th partition. Hence, we can equivalently denote the codeword $y^n$ by $\calB_Y(Q_X,i,j)$. Similarly, for every $y^n\in\calB_Y(Q_X)$, we divide $\calB_Z(y^n)$ into $\lceil|\calB_Z(y^n)/2^{nr_2}|\rceil$ bins, each of size $2^{nr_2}$, except for possibly the last one. Similar to $\calB_Y(Q_X,i,\cdot)$ and $\calB_Y(Q_X,i,j)$, we define $\calB_Z(y^n,u,\cdot)$ and $\calB_Z(y^n,u,v)$. For each $x^n\in\calT_{Q_X}^n$, we use $i_{x^n}$ to denote the index of the partition containing the codeword  associated with $x^n$ in $\calB_Y(Q_X)$ and $j_{x^n}$ to denote the index of the codeword within the partition. Thus, the codeword is denoted as $\calB_Y(Q_X,i_{x^n},j_{x^n})$. Furthermore, given $y^n$, we use $u_{x^n,y^n}$ and $v_{x^n,y^n}$ to denote the corresponding indices of $\calB_Z(y^n)$.  For simplicity, we use $u_{x^n}$ and $v_{x^n}$ to denote $u_{x^n,y^n}$ and $v_{x^n,y^n}$, respectively. Such a notation is valid since $y^n$ is determined given $x^n$. Thus, for a source sequence $x^n$, given a codeword $y^n$, a codeword $z^n$ can be denoted by $\calB_Z(y^n,u_{x^n},v_{x^n})$. Finally, let $m_1(Q_X,i,j)$ be a message from encoder $f_1$ consisting following parts:
\begin{itemize}
\item $\lceil\log|\calQ_{\calX}^n|\rceil$ bits to describe the type $Q_X$.
\item $\lceil\log|\calB_Y(Q_X)/2^{nr_1}|\rceil$ bits to describe index $i$, where $i\in\big[\lceil|\calB_Y(Q_X)/2^{nr_1}|\rceil\big]$.
\item $\lceil\log|\calB_Y(Q_X,i,\cdot)|\rceil$ bits to describe the index $j$, where $j\in\big[\exp\lceil\log{|\calB_Y(Q_X,i,\cdot)|\rceil}\big]$.
\end{itemize}
Similarly, given $y^n$, let $m_2(u,v)$ be a message from encoder $f_2$ consisting following two parts:
\begin{itemize}
\item $\lceil\log|\calB_Z(y^n)/2^{nr_2}|\rceil$ bits to describe index $u$, where $u\in\big[\lceil|\calB_Z(y^n)/2^{nr_2}|\rceil\big]$.
\item $\lceil\log|\calB_Z(y^n,u,\cdot)|\rceil$ bits to describe the index $v$, where $v\in\big[\exp\lceil\log{|\calB_Z(y^n,i,\cdot)|\rceil}\big]$.
\end{itemize}

\subsection{Coding Scheme and Reliability Analysis}\label{sec:DMS_ach_codingscheme}
For any $\delta\in\mathbb{R}$, let
\begin{align}
\calQ(\alpha,\delta)&:=\{Q_X\in\calP(\calX):D(Q_X||P)\leq\alpha+\delta\},\\
\calQ_n(\alpha,\delta)&:=\{Q_X\in\calQ_\calX^n: D(Q_X||P)\leq\alpha+\delta\}.
\end{align}
Then let $\delta>0$ be such that $\max_{Q_X\in\calQ(\alpha,\delta)}R(Q_X,D_1)<R_1$.

Finally, for each sequence $x^n$, let $s_1(x^n)=\lceil\log|\calB_Y(Q_X,i_{x^n},\cdot)|\rceil$ and let $K_{s_1(x^n)}$ be the first $s_1(x^n)$ bits of $K_1^n$. Note that we can denote $s_1(x^n)$ by $s_1(Q_{x^n},i_{x^n})$ since $s_1(x^n)$ depends only on the type and the index of the bin. Furthermore, given $y^n\in\calB_Y(Q_X)$, let $s_2(x^n,y^n)=\lceil\log|\calB_Z(y^n,u,\cdot)|\rceil$ and let $K_{s_2(x^n,y^n)}$ be the first $s_2(x^n,y^n)$ bits of $K_2^n$. Then, we can denote $s_2(x^n,y^n)$ by $s_2(Q_{x^n},i_{x^n},j_{x^n},u_{x^n})$ since $s_2(x^n,y^n)$ depends only on codeword $y^n$ and the index of the bin of $\calB_Z(y^n)$, where $y^n$ is determined by $Q_{x^n},i_{x^n}$ and $j_{x^n}$. The encoders $f_1$ and $f_2$ operate as follows. Given $x^n$, if $Q_{x^n}\in\calQ_n(\alpha,\delta)$,
\begin{align}
f_1(x^n,K_1^n)&=m_1(Q_{x^n},i_{x^n},j_{x^n}\oplus K_{s_1(Q_{x^n},i_{x^n})}),\label{DMSach:def_f1}\\
f_2(x^n,y^n,K_2^n)&=m_2(u_{x^n},v_{x^n}\oplus K_{s_2(Q_{x^n},i_{x^n},j_{x^n},u_{x^n})})\label{DMSach:def_f2},
\end{align}
where the XOR-operation is performed bitwise.

The decoder $\phi_1$ and $\phi_2$ reconstruct the source sequence as follows:
\begin{align}
\phi_1(M_1,K_1^n)&=\calB_Y(Q_X,i_{x^n},j_{x^n}), \label{DMSach:def_phi1}\\
\phi_2(M_1,M_2,K_1^n,K_2^n)&=\calB_Z(y^n,u_{x^n},v_{x^n}).\label{DMSach:def_phi2}
\end{align}
In this case, decoder $\phi_1$ retrieves the type of source sequence and the index of the bin from the first two parts of the message $m_1$, then the index within the bin using the last part of $m_1$ and the key $K_1^n$. Then, decoder $\phi_2$ operate as follows: i) decodes the information of $y^n$ from $M_1$ and $K_1^n$, e.g., $Q_{x^n},i_{x^n}$ and $j_{x^n}$ and chooses the codebook $\calB_Z(y^n)$; ii) retrieves the index of the bin from the first part of the message $m_2$; iii) retrieves the index within the bin from the second part of $m_2$ and the key $K_2^n$. We summarize useful notations in Table~\ref{table:useful_notation}.

\begin{table}
\centering
\caption{Useful Notations}
\begin{tabular}{|c|c|}
\hline
Notation & Description\\
\hline
$\calB_Y(Q_X)$ & Codebook for encoder $f_1$ given type $Q_X$   \\
\hline
$\calB_Z(y^n)$
& Codebook for encoder $f_2$ given a codeword $y^n$ of encoder $f_1$ \\
\hline
$i_{x^n}$ & Bin index given $x^n$ used by encoder $f_1$\\
\hline
$j_{x^n}$ & Index within a bin given $x^n$ used by encoder $f_1$\\
\hline
$u_{x^n,y^n}$ & Bin index given $x^n$ and $y^n$ used by encoder $f_2$, denoted by $u_{x^n}$ for simplicity\\
\hline
$v_{x^n,y^n}$ & Index within a bin given $x^n$ and $y^n$ used by encoder $f_2$, denoted by $v_{x^n}$ for simplicity\\
\hline
$s_1(x^n)$ & $s_1(x^n)=\lceil\log|\calB_Y(Q_X,i_{x^n},\cdot)|\rceil=s_1(Q_{x^n},i_{x^n})$\\
\hline
$s_2(x^n,y^n)$  & $s_2(x^n,y^n)=\lceil\log|\calB_Z(y^n,u,\cdot)|\rceil$  $=s_2(Q_{x^n},i_{x^n},j_{x^n},u_{x^n})$\\
\hline
$K_{s_1(x^n)}$ & first $s_1(x^n)$ bits of $K_1^n$\\
\hline
$K_{s_2(x^n,y^n)}$ & first $s_2(x^n,y^n)$ bits of $K_2^n$\\
\hline
\end{tabular}\label{table:useful_notation}
\end{table}

Now, consider an $m_0\in\calM_1$ that has not been used yet. Note that the requirement of $R_1$ and $R_2$ in \eqref{def:DMS_require_R1} and \eqref{def:DMS_require_R1R2} and the choice of $\delta$ ensures the existence of such $m_0$. For all $x^n$ such that $Q_{x^n}\notin\calQ_n(\alpha,\delta)$,
\begin{align}
f_1(x^n,K_1^n)=f_2(x^n,y^n,K_2^n)=m_0. \label{DMSach:transmit_m0}
\end{align}

To prove that our coding scheme satisfies the JEP constraint, we find that an error occurs if the type of the source sequence is deviated too much from the source distribution, i.e.,
\begin{align}
\rmP_e^n(D_1,D_2)&\leq\sum\limits_{Q_X\notin\calQ_n(\alpha,\delta)}P^n(\calT_{Q_X}^n) \label{DMSach:JEP_scheme}\\
&\leq\sum\limits_{Q_X\notin\calQ_n(\alpha,\delta)}2^{-nD(Q_X\parallel P)} \label{DMSach:JEP_using_type_prob}\\
&\leq(n+1)^{|\calX|}2^{-n(\alpha+\delta)} \label{DMSach:JEP_using_type_number}\\
&\leq2^{-n\alpha} \label{DMSach:JEP_large_n},
\end{align}
where Eq. \eqref{DMSach:JEP_scheme} follows from the design of our coding scheme, Eq. \eqref{DMSach:JEP_using_type_prob} follows from the upper bound for the probability of a type class~\cite[Lemma 2.6]{csiszar2011information}, Eq. \eqref{DMSach:JEP_using_type_number} follows by type counting lemma~\cite[Lemma 2.2]{csiszar2011information} that upper bounds the number of types, and Eq. \eqref{DMSach:JEP_large_n} follows for large enough $n$.

\subsection{Maximal Leakage Analysis}\label{sec:DMS_ach_leakage}
To analyze the maximal leakage of the first layer of encoder and decoder, note that we are leaking the first two parts of message $M_1$, that is, $Q_{X^n}$ and $i_{X^n}$, and hiding the last part $j_{X^n}$. The first part doesn't affect the normalized leakage since there are only polynomial many types. The second part consists roughly of $R(Q,D_1)-r_1$ bits for $R(Q,D_1)>r_1$ and otherwise, for $R(Q,D_1)\leq r_1$, there is no information to be leaked since there is only one bin. The analysis of both two layers of encoders and decoders is similar to the proof of~\cite[Theorem 8]{issa2020leakage}. The eavesdropper receives both message $M_1$ and $M_2$ and retrieves the information of $Q_{X^n},i_{X^n}$ and $u_{X^n}$, whose sum rate is roughly $R(Q,R_1,D_1,D_2)-r_1-r_2$ for $R(Q,R_1,D_1,D_2)>r_1+r_2$. If $R(Q,R_1,D_1,D_2)\leq r_1+r_2$, each message $M_1$ and $M_2$ consists only one bin and there is no information to be leaked.

Let the joint probability distribution of $(x^n,m_1,m_2)$ be $P_{f_1f_2}(x^n,m_1,m_2)$
induced by the source distribution and the stochastic mapping of encoders. Hence, we use $P_{f_1}(m_1|x^n)$ and $P_{f_2}(m_2|x^n,m_1)$ to denote the conditional probability distributions induced by $P_{f_1f_2}(x^n,m_1,m_2)$. Then, given $x^n$ satisfying $Q_{x^n}\in\calQ_n(\alpha,\delta)$, invoking Eq. \eqref{DMSach:def_f1} and Eq. \eqref{DMSach:def_f2}, noting that the key $K_1^n$ and $K_2^n$ are uniformly distributed, it follows that
\begin{align}
P_{f_1}(m_1(Q_{x^n},i_{x^n},j)|x^n)&=2^{-s_1(Q_{x^n},i_{x^n})}, \label{DMSach:uniform_f1}\\
P_{f_2}(m_2(u_{x^n},v)|x^n,Q_{x^n},i_{x^n},j_{x^n})&= 2^{-s_2(Q_{x^n},i_{x^n},j_{x^n},u_{x^n})}. \label{DMSach:uniform_f2}
\end{align}

Similar to~\cite[Section IV. D, Eq.(42)]{issa2020leakage}, except that: i) the distortion level is replaced from $D$ to $D_1$ and the key rate is changed from $r$ to $r_1$, we have that
\begin{align}
\frac{1}{n}\rmL(X^n\to M_1)\leq\max_{Q:D(Q||P)\leq\alpha}\big\{R(Q,D_1)-r_1\big\}^+.
\end{align}

Then we mainly focus on $\rmL(X^n\to M_1M_2)$. Define the following sets
\begin{align}
\calD_0&:=\{x^n\in\calT_{Q_X}^n:Q_X\notin\calQ_n(\alpha,\delta)\},\\
\calD&:=\{x^n\in\calT_{Q_X}^n:Q_X\in\calQ_n(\alpha,\delta)\}.
\end{align}
Recall the definition of maximal leakage in Definition \ref{def:maximal_leakage_DMS}, it follows that
\begin{align}
\nn&\exp\big\{\rmL(X^n\to M_1M_2)\big\}\\*
&\quad=\sum_{(m_1,m_2)\in\calM_1\times\calM_2} \max_{x^n\in\calX^n}P_{f_1,f_2}(m_1,m_2|x^n)\label{DMSach:use_maxLdef}\\
&\quad=\max_{x^n\in\calD_0} P_{f_1,f_2}(m_0,m_0|x^n)+\sum_{\substack{(m_1,m_2)\in\calM_1\times\calM_2,\\ m_1\neq m_0,m_2\neq m_0}} \max_{x^n\in\calD}P_{f_1,f_2}(m_1,m_2|x^n) \label{DMSach:split_messages}\\
\nn&\quad=1+\sum_{Q_X\in\calQ_n(\alpha,\delta)} \sum_{i=1}^{\lceil|\calB_Y(Q_X)/2^{nr_1}|\rceil} \sum_{j=1}^{2^{s_1(Q_X,i)}}\max_{x^n\in\calD}P_{f_1}(m_1(Q_X,i,j)|x^n) \sum_{u=1}^{\lceil\calB_Z(Q_X,i,j)/2^{nr_2}\rceil}\\*
&\qquad\qquad\times\sum_{v=1}^{2^{s_2(Q_X,i,j,u)}}\max_{x^n\in\calD}P_{f_2}(m_2(u,v)|x^n,Q_X,i,j) \label{DMSach:use_codingscheme}\\
&\quad=1+\sum_{Q_X\in\calQ_n(\alpha,\delta)}\sum_{i=1}^{\lceil|\calB_Y(Q_X)/2^{nr_1}|\rceil} \sum_{j=1}^{2^{s_1(Q_X,i)}}2^{-s_1(Q_X,i)} \sum_{u=1}^{\lceil\calB_Z(Q_X,i,j)/2^{nr_2}\rceil}\sum_{v=1}^{2^{s_2(Q_X,i,j,u)}} 2^{-s_2(Q_X,i,j,u)} \label{DMSach:use_uniformkey}\\
&\quad=1+\sum_{Q_X\in\calQ_n(\alpha,\delta)}\sum_{i=1}^{\lceil|\calB_Y(Q_X)/2^{nr_1}|\rceil} \sum_{u=1}^{\lceil\calB_Z(Q_X,i,j)/2^{nr_2}\rceil}\label{DMSach:sum_inner}\\
&\quad\leq1+\sum_{Q_X\in\calQ_n(\alpha,\delta)}(2^{n\{R(Q_X,D_1)+\varepsilon_1-r_1\}^+}+1) (2^{n\{R(Q_X,R_1,D_1,D_2)-R(Q_X,D_1)+\varepsilon_2-r_2\}^+}+1)\label{DMSach:using_lemmaSR}\\
&\quad\leq1+4\sum_{\substack{Q_X\in\\\calQ_n(\alpha,\delta)}}  \exp\Big\{ n\{\{R(Q_X,R_1,D_1,D_2)-R(Q_X,D_1)+\varepsilon_2-r_2\}^+ +\{R(Q_X,D_1)+\varepsilon_1-r_1\}^+\}\Big\} \label{DMSach:upper_2x2y}\\
&\quad\leq1+4\sum_{\substack{Q_X\in\\\calQ_n(\alpha,\delta)}} \exp\Big\{n\max_{Q_X\in\calQ_n(\alpha,\delta)}\big\{ \{R(Q_X,R_1,D_1,D_2)-R(Q_X,D_1)+\varepsilon_2-r_2\}^+ +\{R(Q_X,D_1)+\varepsilon_1-r_1\}^+\big\}\Big\}\label{DMSach:add_max}\\
&\quad\leq1+4\sum_{Q_X\in\calQ_{\calX}^n} \exp\Big\{n\max_{Q_X\in\calQ_n(\alpha,\delta)}\big\{ \{R(Q_X,R_1,D_1,D_2)-R(Q_X,D_1)+\varepsilon_2-r_2\}^+ +\{R(Q_X,D_1)+\varepsilon_1-r_1\}^+\big\}\Big\}\\
&\quad\leq 8(n+1)^{|\calX|}\exp\Big\{ n \max_{Q_X\in\calQ_n(\alpha,\delta)}\big\{\{R(Q_X,R_1,D_1,D_2) -R(Q_X,D_1)+\varepsilon_2-r_2\}^++\{R(Q_X,D_1)+\varepsilon_1-r_1\}^+\big\}\Big\}, \label{DMSach:L_usingtype}
\end{align}
where Eq. \eqref{DMSach:split_messages} follows from Eq. \eqref{DMSach:transmit_m0}, Eq. \eqref{DMSach:use_codingscheme} follows from our coding scheme, Eq. \eqref{DMSach:use_uniformkey} follows from Eq. \eqref{DMSach:uniform_f1} and \eqref{DMSach:uniform_f2}, Eq. \eqref{DMSach:using_lemmaSR} follows from Lemma \ref{theo:lemma_SR_DMS}, with $\calB_Z(y^n)=\calB_Z(Q_X,i,j)$ and the fact that $\lceil|\calB_Y(Q_X)/2^{nr_1}|\rceil$ and $\lceil\calB_Z(Q_X,i,j)/2^{nr_2}\rceil$ are lower bounded by 1, Eq. \eqref{DMSach:upper_2x2y} follows since $x+1\leq 2x$ for $x\geq1$, Eq. \eqref{DMSach:L_usingtype} follows from the type counting lemma~\cite[Lemma 2.2]{csiszar2011information}.

Taking $n\to\infty$, noting that $\varepsilon_1$, $\varepsilon_2$ and $\delta$ are arbitrary and invoking the continuity of $R(Q,R_1,D_1,D_2)$, e.g., \\ $\lim_{\delta\to0}\max_{Q\in\calQ(\alpha,\delta)}R(Q,R_1,D_1,D_2) =\max_{Q\in\calQ(\alpha,0)}R(Q,R_1,D_1,D_2)$ (follows from the convexity of $D(P||Q)$), we have that
\begin{align}
\frac{1}{n}\rmL(X^n\to M_1M_2)\leq\max_{Q:D(Q||P)\leq\alpha}\Big\{ \{R(Q,D_1)-r_1\}^+ +\{R(Q,R_1,D_1,D_2)-R(Q,D_1)-r_2\}^+\Big\},
\end{align}
which completes the proof.

%\vspace{0.5 cm}
\section{Converse Proof of Theorem \ref{theo:DMS}} \label{sec:Conv_TheoremDMS}
To prove the converse part, we need to derive a lower bound of the normalized maximal leakage between the source sequence $X^n$ and the messages $M_1$ and $M_2$. Firstly, inspired by~\cite[Section IV-E]{issa2020leakage}, we propose a guessing scheme of Eve. We next generalize~\cite[Lemma 5]{issa2017guess} to the successive refinement setting in Lemma \ref{theo:lemma_DMS_converse}, which characterizes the probability of correctly guessing the source sequence by Eve. Finally, we derive lower bounds of the normalized maximal leakage to Eve, where the JEP constraint is satisfied by the conditions of Lemma \ref{theo:lemma_DMS_converse}. We find our result is tight under mild conditions since the lower bounds of normalized maximum leakage in proposed guessing scheme coincides with that of upper bounds in our achievability analysis, i.e., our proposed encoding scheme prevents Eve from acquiring more information than the achievability bounds, even if Eve can potentially acquire more information through a better guessing scheme.

\subsection{Guessing Scheme for Eve}\label{sec:Conv_DMS_guessing_scheme}
Consider the following process. The eavesdropper Eve is interested in a randomized function of $X$ called $U$. We assume that $U$ is discrete but unknown to the system designer, which models the fact that we don't know Eve's function of interest. To measure information leakage, we use the following equivalent definition of maximal leakage~\cite[Definition 1]{issa2020leakage}.
\begin{definition}\label{def:maximal_leakage_orginal}
Given a joint distribution $P_{XY}$ on alphabets $\calX$ and $\calY$, the maximal leakage from $X$ to $Y$ is defined as
\begin{align}
\rmL(X\to Y)=\sup_{U-X-Y-\hatU}\log\frac{\Pr\left\{U=\hatU\right\}}{\max_{u\in\calU}P_U(u)}, \label{def:maximal_leakage_original_eq}
\end{align}
where the Markov chain $U-X-Y-\hatU$ holds and the supremum is over all $U$ and $\hatU$ taking values in the same finite, but arbitrary, alphabet.
\end{definition}

Eve observes outputs $M_1$ and $M_2$ of both encoders, and tries to guess $U$ that achieves the supremum in Eq. \eqref{def:maximal_leakage_original_eq} and can verify his/her guess. Consider the guessing scheme of Eve. The adversary Eve first tries to guess the keys $K_i^n$ randomly and uniformly from $\{0,1\}^{nr_i}$ and we denote the guess by $\tilde{K}_i^n$. Then, by assuming that the guess of key was correct, Eve attempts to guess the sequence $x^n$ by using a guessing function $g_i$ given by Lemma \ref{theo:lemma_DMS_converse} below. Finally, again by assuming that the guess of $x^n$ was correct, eve tries to guess $U$ by MAP rule. We denote this stage by $g_U$. Similar to~\cite[Eq. (43)]{issa2020leakage}, we have that
\begin{align}
\Pr\{g_U(x^n)=U|x^n\}=\frac{p^*}{P(x^n)}, \label{DMSconv:result_gU}
\end{align}
where $p^*=\max_{u\in\calU}P_U(u)$.

\subsection{Probability of Correctly Guessing by Eve}\label{sec:Conv_DMS_proof_lemma_converse_DMS}
\begin{lemma}\label{theo:lemma_DMS_converse}
The random function $g_1$ and $g_2$ satisfy that:
\begin{itemize}
\item[i)] There exists a function $g_1:\hat{\calX}_1^n\to\calX^n$ such that for all $(x^n,\hatx_1^n)$ satisfying $d_1(x^n,\hatx_1^n)\leq D_1$,
    \begin{align}
    \Pr\{x^n=g_1(\hatx_1^n)\}\geq b_1^n2^{-n(H_{Q_{x^n}}(X)-R(Q_{x^n},D_1))},
    \end{align}
    where $b_1^n=(n+1)^{-|\calX||\hat{\calX_1}|(|\calX|+1)}$.
\item[ii)] There exist a function $g_2:\hat{\calX}_1^n\times\hat{\calX}_2^n\to\calX^n$ such that for all $(x^n,\hatx_1^n,\hatx_2^n)$ satisfying $d_1(x^n,\hatx_1^n)\leq D_1$, $d_2(x^n,\hatx_2^n)\leq D_2$, $R_1\geq R(Q_{x^n},D_1)$,
    \begin{align}
    \Pr\{x^n=g_2(\hatx_1^n,\hatx_2^n)\} \geq b_2^n2^{-n(H_{Q_{x^n}}(X)-R(Q_{x^n},R_1,D_1,D_2))},
    \end{align}
    where $b_2^n=(n+1)^{-|\calX||\hat{\calX_1}||\hat{\calX_2}|}$.
\end{itemize}
\end{lemma}
\begin{IEEEproof}
The proof of Claim i) is similar to the proof of~\cite[Lemma 12]{issa2020leakage} except that the distortion level is replaced from $D$ to $D_1$. Thus, we mainly focus on the proof of Claim ii), which generalizes the proof of Lemma 5 in~\cite{issa2017guess} to successive refinement setting.

We propose a two-stage scheme for Eve. In the first stage, Eve tries to guess a joint type of $x^n$, $\hatx_1^n$ and $\hatx_2^n$ by observing $\hatx_1^n$ and $\hatx_2^n$. Specifically, Eve chooses an element uniformly at random from the set $\calQ_{\calX\hat\calX_1\hat\calX_2}^n(Q_{\hatx_1^n,\hatx_2^n},D_1,D_2,R_1)$, where
\begin{align}
\nn&\calQ_{\calX\hat\calX_1\hat\calX_2}^n(Q_{\hatx_1^n,\hatx_2^n},D_1,D_2,R_1):=\big\{P_{X\hatX_1\hatX_2} \in\calQ_{\calX\hat\calX_1\hat\calX_2}^n:\\*
&\qquad\quad P_{\hatX_1\hatX_2}=Q_{\hatx_1^n,\hatx_2^n}, \mathbf{E}_{P_{X\hatX_1}}[d_1(X,\hatX_1)\leq D_1], \mathbf{E}_{P_{X\hatX_2}}[d_2(X,\hatX_2)\leq D_2],R_1> R(P_X,D_1)\big\},
\end{align}
where $\calQ_{\calX\hat\calX_1\hat\calX_2}^n$ is the set of types in $\calX\times\hat\calX_1\times\hat\calX_2$. We denote the corresponding function of the this stage by $g_2':\hat\calX_1^n\times\hat\calX_2^n\to \calQ_{\calX\hat\calX_1\hat\calX_2}^n$.

Proceeding by assuming $g'_2(\hatx_1^n,\hatx_2^n)$ is correct joint type, Eve then chooses a sequence uniformly at random from the type class of $Q_{X\hatX_1\hatX_2}^n$. We denote  corresponding function of the this stage by $g_2'':\hat\calX_1^n\times\hat\calX_2^n\times\calQ_{\calX\hat\calX_1\hat\calX_2}^n\to\calX^n$.

Noting that $g_2(\hatx_1^n,\hatx_2^n)=g_2''(\hatx_1^n,\hatx_2^n,g_2'(\hatx_1^n,\hatx_2^n))$, we have that
\begin{align}
\Pr\{x^n=g_2(\hatx_1^n,\hatx_2^n)\}&=\sum_{Q'_{X\hatX_1\hatX_2}\in \calQ_{\calX\hat\calX_1\hat\calX_2}^n(Q_{\hatx_1^n,\hatx_2^n},D_1,D_2,R_1)} \Pr\{g'_2(\hatx_1^n,\hatx_2^n)=Q_{x^n\hatx_1^n\hatx_2^n}\} \Pr\{x^n=g_2''(\hatx_1^n\hatx_2^n,Q_{x^n\hatx_1^n\hatx_2^n})\}\\
&\geq(n+1)^{-|\calX||\hat{\calX_1}||\hat{\calX_2}|} \Pr\{x^n=g_2''(\hatx_1^n\hatx_2^n,Q_{x^n\hatx_1^n\hatx_2^n})\}\label{DMSconv:use_g2'}\\
&\geq(n+1)^{-|\calX||\hat{\calX_1}||\hat{\calX_2}|}2^{-nH(X|\hatX_1,\hatX_2)}\label{DMSconv:use_g2''},
\end{align}
where Eq. \eqref{DMSconv:use_g2'} and Eq. \eqref{DMSconv:use_g2''} follows from the method of types.
Note that
\begin{align}
H(X|\hatX_1,\hatX_2)&=H(X)-H(X)+H(X|\hatX_1,\hatX_2)\\
&=H(X)-I(X;\hatX_1,\hatX_2)\\
&\leq H(X)-R(Q_X,R_1,D_1,D_2), \label{DMSconv:use_def_R(QR1D1D2)}
\end{align}
where Eq. \eqref{DMSconv:use_def_R(QR1D1D2)} follows from the definition of $R(Q_X,R_1,D_1,D_2)$ in Eq. \eqref{def:R(QR1D1D2)}.

The proof is done by combining the results of Eq. \eqref{DMSconv:use_g2''} and Eq. \eqref{DMSconv:use_def_R(QR1D1D2)}.
\end{IEEEproof}

\subsection{Lower Bound Maximal Leakage}\label{sec:Conv_DMS_lowerbounds_maximalleakage}
Let $P_f$ denote the joint distribution of $(X^n,K_1^n,K_2^n,M_1,M_2)$ induced by the source distribution, the keys' distributions and the distributions of messages. Thus, $P_f(M_1,M_2|X^n,K_1,K_2)$ is the induced conditional distribution.

Note that the decoding function $\phi_1$ is a deterministic function of $M_1$ and $K_1^n$, and $\phi_2$ is a deterministic function of $M_1$, $M_2$, $K_1^n$ and $K_2^n$. Let
\begin{align}
\nn&\calM_{D_1D_2}(x^n,k_1,k_2):=\big\{(m_1,m_2)\in\calM_1\times\calM_2: d_1(x^n,\phi_1(m_1,k_1))\leq D_1,\;d_2(x^n,\phi_2(m_1,m_2,k_1,k_2))\leq D_2\big\},\\*
&\qquad\qquad\qquad\qquad\qquad x^n\in\calX^n,k_1\in\calK_1^n,k_2\in\calK_2^n,
\end{align}
and
\begin{align}
\nn\calA:=&\big\{x^n\in\calX^n:d_1(x^n,\phi_1(m_1,k_1))> D_1\;\mathrm{or}\; d_2(x^n,\phi_2(m_1,m_2,k_1,k_2))> D_2\big\},\\ &\mathrm{for\;some}\;\;\;(m_1,m_2)\in\calM_1\times\calM_2,k_1\in K_1^n,k_2\in K_2^n. \label{def:A}
\end{align}

Similar to~\cite[Section IV. E, Eq.(46)-(47)]{issa2020leakage}, except that: i) the distortion level is replaced from $D$ to $D_1$ and the key rate is changed from $r$ to $r_1$, we have that
\begin{align}
\frac{1}{n}\rmL(X^n\to M_1)\geq\max_{Q:D(Q||P)\leq\alpha}\big\{R(Q_X,D_1)-r_1\big\}^+.
\end{align}

We next analyze $\rmL(X^n\to M_1M_2)$. To that end, letting $g$ be the concatenation of all stages, i.e.,\\ $g(M_1,M_2):=g_U(g_2(\phi_2(M_1,M_2,\tilde{K}_1^n,\tilde{K}_2^n)))$, we have that
\begin{align}
\nn&\Pr\big\{U=g(M_1,M_2)\big\}\\*
\nn&\quad=\sum_{x^n\in\calX^n}\sum_{u\in\calU}\sum_{k_1\in\calK_1^n}\sum_{k_2\in\calK_2^n} \sum_{(m_1,m_2)\in\calM_1\times\calM_2}P(x^n) P_{U|X^n}(u|x^n)P_{K_1^n}(k_1)P_{K_2^n}(k_2)P_f(m_1,m_2|x^n,k_1,k_2)\\*
&\quad\quad\cdot\Pr\big\{u=g(m_1,m_2)|x^n,m_1,m_2,k_1,k_2\big\} \label{DMSconv:use_def_g(M1M2)} \\
\nn&\quad\geq\sum_{x^n\in\calX^n}\sum_{u\in\calU}\sum_{k_1\in\calK_1^n}\sum_{k_2\in\calK_2^n}\; \sum_{(m_1,m_2)\in\calM_{D_1D_2}(x^n,k_1,k_2)}P(x^n) P_{U|X^n}(u|x^n)P_{K_1^n}(k_1)P_{K_2^n}(k_2)P_f(m_1,m_2|x^n,k_1,k_2)\\*
&\quad\quad\cdot\Pr\big\{u=g(m_1,m_2)|x^n,m_1,m_2,k_1,k_2\big\} \cdot1\big\{R_1\geq R(Q_X,D_1)\big\}\\
\nn&\quad\geq\sum_{x^n\in\calX^n}\sum_{u\in\calU}\sum_{k_1\in\calK_1^n} \sum_{k_2\in\calK_2^n}\;\sum_{(m_1,m_2)\in\calM_{D_1D_2}(x^n,k_1,k_2)}P(x^n)  P_{U|X^n}(u|x^n)P_{K_1^n}(k_1)P_{K_2^n}(k_2)P_f(m_1,m_2|x^n,k_1,k_2)\\*
&\quad\quad\cdot\Pr\{\tilde{K}_1^n=k_1\}\Pr\{\tilde{K}_2^n=k_2\}\Pr\{g_U(x^n)=u|x^n\} \Pr\{x^n=g_2(\phi_2(m_1,m_2,k_1,k_2))\}1\big\{R_1\geq R(Q_X,D_1)\big\} \label{DMSconv:use_g2gu} \\
\nn&\quad\geq b_2^n\sum_{x^n\in\calX^n}\sum_{k_1\in\calK_1^n}\sum_{k_2\in\calK_2^n}\; \sum_{(m_1,m_2)\in\calM_{D_1D_2}(x^n,k_1,k_2)}P(x^n)  P_{K_1^n}(k_1)P_{K_2^n}(k_2)P_f(m_1,m_2|x^n,k_1,k_2)\cdot2^{-nr_1}2^{-nr_2}\\*
&\quad\quad\cdot 2^{-n(H_{Q_{x^n}}(X)-R(Q_{x^n},R_1,D_1,D_2))}\cdot p^*/P(x^n)\label{DMSconv:use_lemma_g2} \\
\nn&\quad= b_2^n p^*2^{-nr_1}2^{-nr_2}\sum_{Q_X\in\calQ_{\calX}^n}\sum_{x^n\in T_{Q_X}}\sum_{k_1\in\calK_1^n}\sum_{k_2\in\calK_2^n} \sum_{(m_1,m_2)\in\calM_{D_1D_2}(x^n,k_1,k_2)}P(x^n)P_{K_1^n}(k_1)P_{K_2^n}(k_2)\\*
&\quad\quad\cdot P_f(m_1,m_2|x^n,k_1,k_2)\cdot2^{n(R(Q_X,R_1,D_1,D_2)+D(Q_X||P_X))} \label{DMSconv:having_D(QP)_types} \\
&\quad=b_2^n p^*2^{-nr_1}2^{-nr_2}\sum_{Q_X\in\calQ_{\calX}^n}2^{n(R(Q_X,R_1,D_1,D_2)+D(Q_X||P_X))} P_f(\calA^c\cap \calT_{Q_X}^n), \label{DMSconv:result_simplify_toPf}
\end{align}
where Eq. \eqref{DMSconv:use_def_g(M1M2)} follows from the definition of $g(M_1,M_2)$, Eq. \eqref{DMSconv:use_g2gu} follows from that we refine the guessing scheme $g$ by $g_2$ and $g_U$,  Eq. \eqref{DMSconv:use_lemma_g2} follows from Lemma \ref{theo:lemma_DMS_converse}, Eq. \eqref{DMSconv:result_gU} and that Eve guesses the keys $K_i^n$ randomly and uniformly from $\{0,1\}^{nr_i}$ for $i\in[2]$ and Eq. \eqref{DMSconv:having_D(QP)_types} follows from the method of types. For any $Q_X$, it follows that
\begin{align}
P_f(\calA^c|\calT_{Q_X}^n)
&=1-P_f(\calA|\calT_{Q_X}^n)\\
&\geq 1-\min\left\{1,\frac{P_f(\calA)}{P(\calT_{Q_X}^n)}\right\} \label{DMSconv:explain_P(TQ)} \\
&\geq 1-\min\left\{1,2^{-n(\alpha-D(Q_X||P_X)-\frac{|\calX|}{n}\log(n+1))}\right\}\label{DMSconv:use_JEP}\\
&=\left\{1-2^{-n(\alpha-D(Q_X||P_X)-\frac{|\calX|}{n}\log(n+1))}\right\}^+,\label{DMSconv:result_PfAcTQ}
\end{align}
where the $P(\calT_{Q_X}^n)$ in Eq. \eqref{DMSconv:explain_P(TQ)} denotes the probability of the type class $\calT_{Q_X}^n$ under distribution $P$, Eq. \eqref{DMSconv:use_JEP} follows from the lower bound of the probability of type class~\cite[Lemma 2.6]{csiszar2011information}, $P_f(\calA)=\rmP_{\rme}^n(D_1,D_2)$ (cf. Eq. \eqref{def:JEP} and Eq. \eqref{def:A}) and the fact that $\rmP_{\rme}^n(D_1,D_2)\leq2^{-n\alpha}$.

For simplicity, let $b_3^n=\frac{|\calX|}{n}\log(n+1)$. Combining the results of Eq. \eqref{DMSconv:result_simplify_toPf} and Eq. \eqref{DMSconv:result_PfAcTQ}, fixing $\tau>0$, we have that
\begin{align}
\nn&\Pr\big\{U=g(M_1,M_2)\big\}\\*
&\quad\geq b_2^n p^*2^{-nr_1}2^{-nr_2} \sum_{Q_X\in\calQ_{\calX}^n}2^{n(R(Q_X,R_1,D_1,D_2)+D(Q_X||P_X))} P(\calT_{Q_X}^n)\cdot\left\{1-2^{-n(\alpha-D(Q_X||P_X)-b_3^n)}\right\}^+\\
&\quad\geq b_4^n p^*2^{-nr_1}2^{-nr_2}\sum_{Q_X\in\calQ_{n}(\alpha,-\tau)}2^{nR(Q_X,R_1,D_1,D_2)} (1-2^{-n(\alpha-D(Q_X||P_X)-b_3^n)})\label{DMSconv:use_b4n}\\
&\quad\geq b_4^n p^*2^{-nr_1}2^{-nr_2}\sum_{Q_X\in\calQ_{n}(\alpha,-\tau)}2^{nR(Q_X,R_1,D_1,D_2)}\cdot\frac{1}{2} \label{DMSconv:use_1/2}\\
&\quad\geq\frac{b_4^n p^*}{2}\max_{Q_X\in\calQ_{n}(\alpha,-\tau)}2^{n(R(Q_X,R_1,D_1,D_2)-r_1-r_2)},\label{DMSconv:result_U_gM1M2}
\end{align}
where Eq. \eqref{DMSconv:use_b4n} follows from the lower bound of the probability of type class and $b_4^n=(n+1)^{-|\calX|}b_2^n$, and Eq. \eqref{DMSconv:use_1/2} follows for large enough $n$.

Invoking Definition \ref{def:maximal_leakage_orginal}, we have that
\begin{align}
\frac{1}{n}\rmL(X^n\to M_1M_2)&\geq \frac{1}{n}\log\frac{\Pr\left\{U=g(M_1,M_2)\right\}}{\max_{u\in\calU}P_U(u)}\\
&\geq\frac{1}{n}\log\frac{b_4^n}{2}\max_{Q_X\in\calQ_{n}(\alpha,-\tau)}2^{n(R(Q_X,R_1,D_1,D_2)-r_1-r_2)} \label{DMSconv:use_def_p*}\\
&\geq\max_{Q_X:D(Q_X||P)\leq\alpha}R(Q_X,R_1,D_1,D_2)-r_1-r_2,\label{DMSconv:use_continue}
\end{align}
where Eq. \eqref{DMSconv:use_def_p*} follows from Eq. \eqref{DMSconv:result_U_gM1M2} and the fact that $p^*=\max_{u\in\calU}P_U(u)$ and Eq. \eqref{DMSconv:use_continue} follows from the continuity of $R(Q_X,R_1,D_1,D_2)$. Since $\calL(X^n\to M_1M_2)$ is non-negative by definition, it follows that
\begin{align}
\frac{1}{n}\rmL(X^n\to M_1M_2) \geq\max_{Q:D(Q||P)\leq\alpha}\big\{R(Q,R_1,D_1,D_2)-r_1-r_2\big\}^+. \label{DMSconv:L_positive}
\end{align}

\section{Proof of Theorem \ref{theo:DMS_expectation}}\label{sec:Conv_DMS_expectation}
\subsection{Proof of Main Results}
The proof of the achievability part is similar to the case under JEP in Section \ref{sec:Ach_TheoremDMS}, which can be obtained from the achievability analyses from Theorem \ref{theo:DMS}. Specifically,
\begin{align}
\limsup_{n\to\infty}\frac{1}{n}\rmL(X^n\to M_1)&\leq \lim_{\alpha\to0}\Lambda_1(P,\vec{R_1},D_1,\alpha)\\
&=\Omega_1(P,\vec{R_1},D_1),\\
\limsup_{n\to\infty}\frac{1}{n}\rmL(X^n\to M_1M_2) &\leq\lim_{\alpha\to0}\Lambda_2(P,\vec{R_1},\vec{R_2},D_1,D_2,\alpha)\\
&=\Omega_2(P,\vec{R_1},\vec{R_2},D_1,D_2),
\end{align}
where $\alpha$ can be chosen as $\alpha=\frac{\log n}{n}$, such that JEP exponent is upper bounded by $\frac{1}{n}$. Thus, the expected distortion can be ensured~\cite[pp. 190]{zhou2023monograph} by the rate requirements in Eq. \eqref{DMSexpect:req_R1}-\eqref{DMSexpect:req_R1R2}.

To prove the converse part, we need the following lemma that characterizes the reliability performance under the expected distortion constraint and the secrecy performance under the equivocation constraint. Let $(E_1,E_2)\in\bbR_+^2$ be arbitrary.
\begin{lemma}\label{theo:lemma_DMS_expect_corollary5}
Let $\calR$ denote the closure of pairs $(\vec{R_1},\vec{R_2},D_1,D_2,E_1,E_2)$ such that there exists a sequence of $(n,\vec{R_1},\vec{R_2})$-codes satisfying the expected distortion constraints in \eqref{avg:const1}, \eqref{avg:const2} and the following two secrecy constraints
\begin{align}
\liminf_{n\to\infty}\frac{1}{n}H(X^n|M_1)&\geq E_1,\label{equi:1}\\
\liminf_{n\to\infty}\frac{1}{n}H(X^n|M_1,M_2)&\geq E_2\label{equi:2}.
\end{align}
We have the following outer bound for $\calR$:
\begin{align}
\calR\subseteq\calR_{\mathrm{out}}=\bigcup_{P_{\hatX_1\hatX_2|X}}\left\{
\begin{array}{l}
(\vec{R_1},\vec{R_2},D_1,D_2,E_1,E_2):\\
R_1\geq I(X;\hatX_1)\\
R_1+R_2\geq I(X;\hatX_1,\hatX_2)\\
D_1\geq\mathbf{E}\big[d_1(X,\hatX_1)\big]\\
D_2\geq\mathbf{E}\big[d_2(X,\hatX_2)\big]\\
E_1\leq H(X)-\{I(X;\hatX_1)-r_1\}^+\\
E_2\leq H(X)-\{I(X;\hatX_1,\hatX_2)-r_1-r_2\}^+
\end{array}
\right\}. \label{DMStheo:calR}
\end{align}
\end{lemma}
The proof of Lemma~\ref{theo:lemma_DMS_expect_corollary5} is provided in Section \ref{sec:Conv_DMS_lemma_expected}.

For any discrete random variable, it follows that
\begin{align}
\rmL(X^n\to M_1)&=I_{\infty}(X^n;M_1) \label{DMSexpectconv:equal_infintyI}\\
&\geq I(X^n;M_1), \label{DMSexpectconv:geq_I}
\end{align}
where Eq. \eqref{DMSexpectconv:equal_infintyI} follows from~\cite[Theorem 1]{issa2020leakage} and Eq. \eqref{DMSexpectconv:geq_I} follows from~\cite[Theorem 2]{verdu2015alpha}.

Invoking Lemma~\ref{theo:lemma_DMS_expect_corollary5}, we lower bound $I(X^n;M_1)$ as follows:
\begin{align}
I(X^n;M_1)&=H(X^n)-H(X^n|M_1)\\
&=nH(X)-H(X^n|M_1)\\
&\geq n\{I(X;\hatX_1)-r_1\}^+, \label{DMSexpectconv:using_lemma4}
\end{align}
where Eq. \eqref{DMSexpectconv:using_lemma4} follows since $E_1\leq H(X)-\{I(X;\hatX_1)-r_1\}^+$ and we take equality for Eq. \eqref{equi:1}. Such a choice is valid since the mutual information leakage is minimized when the equivocation of adversary is maximized.
Similarly, we have
\begin{align}
\rmL(X^n\to M_1M_2)&\geq n\{I(X;\hatX_1,\hatX_2)-r_1-r_2\}^+.
\end{align}
The proof is completed using the definitions of $R(P,D_1)$ and $R(P,R_1,D_1,D_2)$ in Eq. \eqref{def:R(QD1)} and \eqref{def:R(QR1D1D2)}, respectively, and dividing $n$ at both side.

\subsection{Proof of Lemma~\ref{theo:lemma_DMS_expect_corollary5}}\label{sec:Conv_DMS_lemma_expected}
The proof of Lemma~\ref{theo:lemma_DMS_expect_corollary5} is decomposed into three steps. Firstly, inspired by~\cite[Theorem 1]{schieler2014}, in Section~\ref{sec:Conv_DMS_expectation_causaldef}, we generalize the causal disclosure problem from the point-to-point setting~\cite[Section II]{schieler2014} to the successive refinement setting. Here causal disclosure means that Eve has additional access to the past source and reconstruction symbols beyond the encoded messages output by both encoders. Subsequently, in Section~\ref{sec:Conv_DMS_expectation_generalcausal}, we derive a converse bound for the successive refinement problem with causal disclosure. Finally, via proper specialization, in Section~\ref{sec:Conv_DMS_expectation_lemma_3}, we derive a converse bound for the successive refinement setting of SCS under the equivocation secrecy constraint as in Lemma \ref{theo:lemma_DMS_expect_corollary5}. The reason why the causal disclosure setting is valid to establish the results in Lemma \ref{theo:lemma_DMS_expect_corollary5} and Theorem~\ref{theo:DMS_expectation} is explained in Section~\ref{sec:Conv_DMS_expectation_why}.

\subsubsection{Causal Disclosure under Successive Refinement}\label{sec:Conv_DMS_expectation_causaldef}
Recall the successive refinement setting of SCS illustrated in Fig.\ref{fig:systemmodel}. In the causal disclosure setting, as illustrated in Fig. \ref{fig:causal_disclosure}, at each time $i\in[n]$, the eavesdropper has additional access to potentially noisy observations of the past source sequence and reconstruction symbols $(W_0^{i-1},W_1^{i-1},W_2^{i-1})$. In particular, $W_{0,i}$ is the output of passing $X_i$ through a noisy channel $P_{W_0|X}$, and $W_{k,i}$ is the output of passing $\hatX_{k,i}$ through a noisy channel $P_{W_i|\hatX_{k,i}}$ for each $k\in[2]$. Using the message $M_1$ and causal observations $(W_0^{i-1},W_1^{i-1})$, Eve aims to estimate a function of the $i$-th source symbol $X_i$ as $C_{1,i}$ using decoder $\psi_{1,i}$. Using the messages $(M_1,M_2)$ and causal observations $(W_0^{i-1},W_1^{i-1},W_2^{i-1})$, Eve aims to estimate another function of the $i$-th source symbol $X_i$ as $C_{2,i}$ using decoder $\psi_{2,i}$. For simplicity, we denote the pair $(W_0^i,W_1^i)$ and $(W_0^i,W_2^i)$ by $W_\alpha^i$ and $W_\beta^i$, respectively.

Let $(\calW_0,\calW_1,\calW_2,\calC_1,\calC_2)$ be finite alphabets. Recall the definition of an $(n,\vec{R_1},\vec{R_2})$-code in Definition~\ref{def:nR1R2_code}, an $(n,\vec{R_1},\vec{R_2})$-causal disclosure code has two more decoders for Eve at each step $i\in[n]$:
\begin{align}
\psi_{1,i}&:\calM_1\times\calW_0^{i-1}\times\calW_1^{i-1}\to\calC_1\\
\psi_{2,i}&:\calM_1\times\calM_2\times\calW_0^{i-1}\times\calW_1^{i-1}\times\calW_2^{i-1}\to\calC_2.
\end{align}

Define the following two symbol-wise pay-off functions for Eve: $\pi_1:\calX\times\hat{\calX}_1\times\calC_1\to(0,\infty)$ and $\pi_2:\calX\times\hat{\calX}_2\times\calC_2\to(0,\infty)$.
\begin{definition}
Fix a source distribution $P$. The quadruple $(\vec{R_1},\vec{R_2},\Pi_1,\Pi_2)$ is achievable if there exists a sequence of $(n,\vec{R_1},\vec{R_2})$-causal disclosure codes such that
\begin{align}
\liminf_{n\to\infty}\min_{\{P_{C_{1,i}|M_1,W_\alpha^{i-1}}\}_{i=1}^n}\mathbf{E}\bigg[\frac{1}{n} \sum_{i=1}^n\pi_1(X_i,\hatX_{1,i},C_{1,i})\bigg] \geq &\Pi_1, \label{def:Pi_1}\\
\liminf_{n\to\infty}\min_{\{P_{C_{2,i}|M_1,M_2,W_\alpha^{i-1},W_\beta^{i-1}}\}_{i=1}^n}\mathbf{E}\bigg[\frac{1}{n} \sum_{i=1}^n\pi_2(X_i,\hatX_{1,i}, \hatX_{2,i},C_{2,i})\bigg] \geq &\Pi_2. \label{def:Pi_2}
\end{align}
The closure of the set of all achievable quadruple $(\vec{R_1},\vec{R_2},\Pi_1,\Pi_2)$ is called optimal achievable $(\vec{R_1},\vec{R_2},\Pi_1,\Pi_2)$ region and denoted as $\calS$.
\end{definition}

{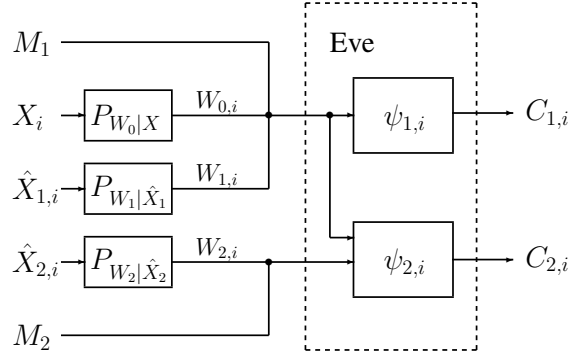
\begin{figure}[t]
\centering
\setlength{\unitlength}{0.5cm}
\scalebox{0.65}{
\begin{picture}(22,18)
\linethickness{1pt}
%causal disclosure
\put(0,12){\makebox{\LARGE$X_i$}}
\put(0,9){\makebox{\LARGE$\hatX_{1,i}$}}
\put(0,6){\makebox{\LARGE$\hatX_{2,i}$}}
\put(2,12.4){\vector(1,0){1}}
\put(2,9.4){\vector(1,0){1}}
\put(2,6.4){\vector(1,0){1}}
\put(3,11.4){\framebox(3.5,2)}
\put(3.2,12){\makebox{\LARGE$P_{W_0|X}$}}
\put(3,8.4){\framebox(3.5,2)}
\put(3.2,9.1){\makebox{\LARGE$P_{W_1|\hatX_1}$}}
\put(3,5.4){\framebox(3.5,2)}
\put(3.2,6.1){\makebox{\LARGE$P_{W_2|\hatX_2}$}}
\put(6.5,12.4){\line(1,0){4}}
\put(6.5,9.4){\line(1,0){4}}
\put(6.5,6.4){\line(1,0){4}}
\put(7.5,12.7){\makebox{\Large$W_{0,i}$}}
\put(7.5,9.7){\makebox{\Large$W_{1,i}$}}
\put(7.5,6.7){\makebox{\Large$W_{2,i}$}}
%Eve
\put(14,5){\framebox(4,3)}
\put(15.2,6.1){\makebox{\LARGE $\psi_{2,i}$}}
\put(14,10.9){\framebox(4,3)}
\put(15.2,12){\makebox{\LARGE $\psi_{1,i}$}}
\multiput(12,17)(0,-0.4){36}{\line(0,-1){0.2}}
\multiput(12,17)(0.4,0){18}{\line(1,0){0.2}}
\multiput(19,17)(0,-0.4){36}{\line(0,-1){0.2}}
\multiput(12,2.8)(0.4,0){18}{\line(1,0){0.2}}
\put(13,15){\makebox{\LARGE Eve}}
%input messages
\put(0,3){\makebox{\LARGE$M_2$}}
\put(2,3.4){\line(1,0){8.5}}
\put(0,15){\makebox{\LARGE$M_1$}}
\put(2,15.4){\line(1,0){8.5}}
%input lines
\put(10.5,15.4){\line(0,-1){6}}
\put(10.5,12.4){\circle*{0.3})}
\put(10.5,6.4){\line(0,-1){3}}
\put(10.5,6.4){\circle*{0.3})}
\put(10.5,12.4){\vector(1,0){3.5}}
\put(10.5,6.4){\vector(1,0){3.5}}
\put(13,12.4){\line(0,-1){5}}
\put(13,12.4){\circle*{0.3})}
\put(13,7.4){\vector(1,0){1}}
%output decoded symbols
\put(18,6.5){\vector(1,0){2.5}}
\put(21,6.2){\makebox{\LARGE $C_{2,i}$}}
\put(18,12.5){\vector(1,0){2.5}}
\put(21,12.2){\makebox{\LARGE $C_{1,i}$}}
\end{picture}}
\caption{Illustration of the causal disclosure setting for Eve.}
\label{fig:causal_disclosure}
\end{figure}

By the second remark after~\cite[Definition 2]{schieler2014}, we can also assume that Eve deploys a set of deterministic decoding functions $\{\psi_{1,i}(m_1,w_\alpha^{i-1})\}_{i=1}^n$ and $\{\psi_{2,i}(m_1,m_2,w_\alpha^{i-1},w_\beta^{i-1})\}_{i=1}^n$.

\subsubsection{General Results under Causal Disclosure Setting}\label{sec:Conv_DMS_expectation_generalcausal}
For $i\in[2]$, let $U_i$ and $V_i$ be random variables taking values in the alphabets $\calU_i$ and $\calV_i$ respectively. Define a set of joint distributions on $\calW_0\times\calW_1\times\calW_2\times\calX\times\hat{\calX_1}\times\hat{\calX_2} \times\calU_1\times\calU_2\times\calV_1\times\calV_2$ as %\red{Omit the repetition one?}
\begin{align}
\nn\calB:=\Big\{&Q_{W_0W_1W_2X\hatX_1\hatX_2U_1U_2V_1V_2}: W_0-X-(U_1,V_1)-\hatX_1-W_1,\; W_0-X-(U_2,V_2)-\hatX_2-W_2,\\
\nn&|\calU_1|\leq |\calX|+5,\;|\calV_1|\leq|\calX||\hat\calX_1|(|\calX|+5)+3, \;|\calU_2|\leq|\calX|(|\calX|+5)(|\calX||\hat\calX_1|(|\calX|+5)+3)+2,\\
\nn&|\calV_2|\leq|\calX||\hat\calX_1||\hat\calX_2|(|\calX|+5) (|\calX||\hat\calX_1|(|\calX|+5)+3) (|\calX|(|\calX|+5)(|\calX||\hat\calX_1|(|\calX|+5)+3)+2)+1\Big\}.
\end{align}
For simplicity, we use $\tilde{Q}$ to denote $Q_{W_0W_1W_2X\hatX_1\hatX_2U_1U_2V_1V_2}$.

For the successive refinement of SCS with causal disclosure, we derive a converse bound on the message rate, key rate of the legitimate users and the payoff of Eve in the following lemma, which extends the converse part of~\cite[Theorem 1]{schieler2014}.
\begin{lemma}\label{theo:lemma_cuff_theo1}
Fix $P_X$ and the causal disclosure channels $P_{W_0|X}$ and $P_{W_i|\hatX_i}$ for $i\in[2]$. The optimal region of  $(\vec{R_1},\vec{R_2},\Pi_1,\Pi_2)$ satisfies that
{\begin{align}
\calS\subseteq\calS_{\mathrm{out}}=\bigcup_{\tilde{Q}\in\calB}\left\{
\begin{array}{l}
(\vec{R_1},\vec{R_2},\Pi_1,\Pi_2):\\
R_1\geq I(X;U_1,V_1)\\
R_1+R_2\geq I(X;U_1,U_2,V_1,V_2)\\
r_1\geq I(W_0,W_1;V_1|U_1)\\
r_1+r_2\geq I(W_0,W_1,W_2;V_1,V_2|U_1,U_2)\\
\Pi_1\leq\min\limits_{h_1\in\calH_1}\mathbf{E}[\pi_1(X,\hatX_1,h_1(U_1))]\\
\Pi_2\leq\min\limits_{h_2\in\calH_2}\mathbf{E}[\pi_2(X,\hatX_1,\hatX_2,h_2(U_1,U_2))]
\end{array}
\right\},\label{DMSexpectconv:region_cuff_theo1}
\end{align}}
where $\calH_1:=\{h_1:\calU_1\to\calC_1\}$ and $\calH_2:=\{h_2:\calU_1\times\calU_2\to\calC_2\}$.
\end{lemma}
\begin{proof}
%\red{We further weaken the conditions of the decoders by allowing $\phi_1$ and $\phi_2$ causal access to source and the eavesdropper, which can be shown that it does not increase the payoff.}
Fix a source distribution $P_X$, payoff functions $\pi_1(x,\hatx_1,c_1)$ and $\pi_2(x,\hatx_1,\hatx_2,c_2)$ and causal disclosure channels $P_{W_0|X}$ and $P_{W_i|\hatX_i}$ for $i\in[2]$. For each $j\in[n]$, we define the following auxiliary variables $U_{1,j}:=(M_1,W_\alpha^{j-1})$ and $U_{2,j}:=(M_2,W_\beta^{j-1})$. Let $J$ be an auxiliary variable distributed uniformly from $[n]$, independently from\\ $(X^n,\hatX_1^n,\hatX_2^n,W_\alpha^n,W_\beta^n,M_1,M_2,K_1,K_2)$.

It follows from \cite[Section VII]{schieler2014} that the following bounds on $R_1$, $r_1$ and $\Pi_1$ hold:
\begin{align}
R_1&\geq nI(X;U_1,V_1), \label{DMSconv:cuff_R1}\\
r_1&\geq I(W_\alpha,W_\beta;V_1|U_1),\label{DMSconv:cuff_r1}\\
\Pi_1&\leq\min_{h_1\in\calH_1}\mathop\mathbf{E} \Big[\pi_2(X,\hatX_1,h_1(U_1))\Big],\label{DMSconv:cuff_Pi1}
\end{align}
where Eq. \eqref{DMSconv:cuff_R1}, Eq. \eqref{DMSconv:cuff_r1} and  Eq. \eqref{DMSconv:cuff_Pi1} follows from Eq. (189)-(198), Eq. (199-204) and Eq. (205)-(208) in~\cite{schieler2014}, respectively.

We now bound the additional terms in the successive refinement setting. The sum rate is lower bounded as follows:
\begin{align}
n(R_1+R_2)&\geq H(M_1,M_2)\\
&\geq H(M_1,M_2|K_1,K_2)\\
&\geq I(X^n;M_1,M_2|K_1,K_2)\\
&=I(X^n;M_1,M_2,K_1,K_2)\label{DMSexpectconv:I(XMK)}\\
&=\sum_{j=1}^nI(X_j;M_1,M_2,K_1,K_2|X^{j-1})\\
&=\sum_{j=1}^nI(X_j;M_1,M_2,K_1,K_2,X^{j-1})\\
&=\sum_{j=1}^nI(X_j;M_1,M_2,K_1,K_2, X^{j-1},W_\alpha^{j-1},W_\beta^{j-1}) \label{DMSexpectconv:I(XMKW)}\\
&=\sum_{j=1}^nI(X_j;U_{1,j},U_{2,j},K_1,K_2,X^{j-1}) \label{DMSexpectconv:I(XU1U2V1V2X^j-1)}\\
&\geq\sum_{j=1}^nI(X_j;U_{1,j},U_{2,j},K_1,K_2) \\
&=nI(X_J;U_{1,J},U_{2,J},K_1,K_2,J),
\end{align}
where Eq. \eqref{DMSexpectconv:I(XMK)} follows since $X^n$ is independent from $K_1$ and $K_2$, Eq. \eqref{DMSexpectconv:I(XMKW)} follows from the Markov chain $X_j-(M_1,M_2,K_1,K_2,X^{j-1})-(W_\alpha^{j-1},W_\beta^{j-1})$, Eq. \eqref{DMSexpectconv:I(XU1U2V1V2X^j-1)} follows from the definitions of $U_{1,j}$ and $U_{2,j}$.

Similarly, we can lower bound the sum key rate as follows:
\begin{align}
n(r_1+r_2)&\geq H(K_1,K_2)\\
&\geq H(K_1,K_2|M_1,M_2)\\
&\geq I(W_\alpha^n,W_\beta^n;K_1,K_2|M_1,M_2)\\
&\geq \sum_{j=1}^nI(W_{\alpha,j},W_{\beta,j};K_1,K_2|M_1,M_2, W_\alpha^{j-1},W_\beta^{j-1})\\
&=\sum_{j=1}^nI(W_{\alpha,j},W_{\beta,j};K_1,K_2|U_{1,j},U_{2,j})\\
&=nI(W_{\alpha,J},W_{\beta,J};K_1,K_2|U_{1,J},U_{2,J},J).
\end{align}

Furthermore, we can upper bound the payoff of Eve as follows:
\begin{align}
\Pi_2&\leq\min_{\psi_2\in\Psi_2} \mathop{\mathbf{E}}\limits_{P_{X\hatX_1\hatX_2M_1M_2W_\alpha W_\beta J}} \bigg[\frac{1}{n}\sum_{j=1}^n\pi_2\big(X_j,\hatX_{1,j}\hatX_{2,j}, \psi_2(M_1,M_2,W_\alpha^{j-1},W_\beta^{j-1},j)\big)\bigg]\\
&\leq\min_{\psi_2\in\Psi_2} \mathop{\mathbf{E}}\limits_{P_{X\hatX_1\hatX_2U_1U_2J}} \bigg[\frac{1}{n}\sum_{j=1}^n\pi_2\big(X_j,\hatX_{1,j}\hatX_{2,j},
\psi_2(U_{1,j},U_{2,j},j)\big)\bigg]\\
&=\min_{\psi_2\in\Psi_2}\mathop\mathbf{E}\limits_{P_J} \bigg[\mathop\mathbf{E}\limits_{P_{X\hatX_1\hatX_2U_{1,J}U_{2,J}|J}} \Big[\pi_2\big(X_J,\hatX_{1,J}\hatX_{2,J},\psi_2(U_{1,J},U_{2,J},J)\big)|J\Big]\bigg] \label{DMSexpectconv:EEpiXpsiJ}\\
&=\min_{\psi_2\in\Psi_2}\mathop\mathbf{E}\limits_{P_{X\hatX_1\hatX_2U_{1,J}U_{2,J}J}} \Big[\pi_2\big(X_J,\hatX_{1,J}\hatX_{2,J},\psi_2(U_{1,J},U_{2,J},J)\big)\Big].
\end{align}

Let $U_1:=(U_{1,J},J)$, $U_2:=(U_{2,J},J)$, $X=X_J$, $\hatX_1=\hatX_{1,J}$, $\hatX_2=\hatX_{2,J}$, $W_\alpha=W_{\alpha,J}$, $W_\beta=W_{\beta,J}$, $V_1=K_1$ and $V_2=K_2$. It follows that
\begin{align}
R_1+R_2&\geq I(X;U_1,U_2,V_1,V_2), \label{DMSexpectconv:I(XU1U2V1V2)}\\
r_1+r_2&\geq I(W_\alpha,W_\beta;V_1,V_2|U_1,U_2), \label{DMSexpectconv:I(WVU)}\\
\Pi_2&\leq \min_{h_2\in\calH_2}\mathop\mathbf{E}\Big[\pi_2(X,\hatX_1,\hatX_2,h_2(U_1,U_2))\Big].
\end{align}
Furthermore, the following Markov chain holds:
\begin{align}
W_0-X-&(U_1,V_1)-\hatX_1-W_1,\\
W_0-X-&(U_2,V_2)-\hatX_2-W_2.
\end{align}
To prove the cardinality bounds for $\calU_1$, $\calU_2$, $\calV_1$ and $\calV_2$, we use the support lemma~\cite[Appendix C]{el2011network}.
\end{proof}

\subsubsection{Final Steps}\label{sec:Conv_DMS_expectation_lemma_3}
Now we can prove Lemma~\ref{theo:lemma_DMS_expect_corollary5} by showing that normalized equivocation-based metric is a special case of causal disclosure of  Lemma~\ref{theo:lemma_cuff_theo1} if one chooses the payoff functions to be log-loss functions.

Firstly, we replace the payoff functions $\pi_1$ in Eq. \eqref{def:Pi_1} and $\pi_2$ in Eq. \eqref{def:Pi_2} by the following distortion functions:
\begin{align}
\tilde{d_1}(x^n,\hatx_1^n,c_1^n)&:=\frac{1}{n}\sum_{i=1}^n\tilde{d_1}(x^n,\hatx_{1,i}^n,c_{1,i}),\\
\tilde{d_2}(x^n,\hatx_1^n,\hatx_2^n,c_2^n) &:=\frac{1}{n}\sum_{i=1}^n\tilde{d_2}(x^n,\hatx_{1,i}^n,\hatx_{2,i}^n, c_{2,i}),
\end{align}
and we use $E_1$ and $E_2$ to replace $\Pi_1$ and $\Pi_2$, respectively.

{
Now choose the causal disclose such that for each $i\in[n]$, $W_{0,i}=X_i$  and $W_{1,i}=W_{2,i}$ equals to an arbitrary constant. Recall that $\calP_{\calX|\calY}$ denotes the set of all conditional probability distributions on an alphabet $\calX$ given another alphabet $\calY$, and $C_i$ is the estimation generated by Eve. We set the distortion functions $\tilde{d_1}$ and $\tilde{d_2}$ to be the following log-loss functions:
\begin{align}
\tilde{d_1}(x_i,\hatx_{1,i},c_{1,i})&=-\log c_{1,i}(x_i|m_1,x^{i-1}),\\
\tilde{d_2}(x,\hatx_1,\hatx_2,c_2)&=-\log c_{2,i}(x_i|m_1,m_2,x^{i-1}),
\end{align}
where $c_{1,i}\in\calP_{\calX|\calM_1,\calX^{i-1}}$ and $c_{2,i}\in\calP_{\calX|\calM_1,\calX^{i-1}}$ denote two conditional distributions that correspond to soft decoding.
Let the joint distribution of $X^n,M_1,M_2,\hatX_1^n,\hatX_2^n$ in the successive refinement setting of SCS be $P(X^n,M_1,M_2,\hatX_1^n,\hatX_2^n)$. In the following analyses, the expectation is calculated with respect to the above joint distribution or its induced distributions. It follows that
\begin{align}
E_1&\leq\min_{c_1^n:\{c_{1,i}\in\calP_{\calX|\calM_1,\calX^{i-1}}\}_{i=1}^n} \mathbf{E}\Big[\frac{1}{n}\sum_{i=1}^{n}\tilde{d_1}(X_i,\hatX_{1,i},c_{1,i})\Big]\\ &=\frac{1}{n}\sum_{i=1}^{n}\min_{c_{1,i}\in\calP_{\calX|\calM_1,\calX^{i-1}}} \mathbf{E}\Big[\tilde{d_1}(X_i,\hatX_{1,i},c_{1,i})\Big]\\
&=\frac{1}{n}\sum_{i=1}^{n}\min_{c_{1,i}\in\calP_{\calX|\calM_1,\calX^{i-1}}} \mathbf{E}\Big[\log\frac{1}{c_{1,i}(X_i|M_1,X^{i-1})}\Big]\\
&=\frac{1}{n}\sum_{i=1}^{n}\min_{c_{1,i}\in\calP_{\calX|\calM_1,\calX^{i-1}}}\Big\{ \mathbf{E}\Big[\log\frac{1}{P(X_i|M_1,X^{i-1})}\Big] +\mathbf{E}\Big[\log\frac{P(X_i|M_1,X^{i-1})}{c_{1,i}(X_i|M_1,X^{i-1})}\Big]\Big\}\\
&=\frac{1}{n}\sum_{i=1}^{n}\min_{c_{1,i}\in\calP_{\calX|\calM_1,\calX^{i-1}}}\Big\{ H(X_i|M_1,X^{i-1})+\sum_{m_1,x^{i-1}}P(M_1,x^{i-1}) D(P(X_i|m_1,x^{i-1})||c_{1,i}(X_i|m,x^{i-1})) \Big\}\\
&=\frac{1}{n}\sum_{i=1}^{n}H(X_i|M_1,X^{i-1}) \label{DMSexpectconv:using_cuff_lemma2}\\
&=\frac{1}{n}H(X^n|M_1),
\end{align}
where Eq. \eqref{DMSexpectconv:using_cuff_lemma2} follows since the minimization is over all $C_1$ and the fact that $X_i-(M_1,X^{i-1})-C_{i}$. Similarly, it follows that
\begin{align}
E_2\leq \frac{1}{n}H(X^n|M_1,M_2).
\end{align}}

Combining the above arguments with the bounds on $\Pi_1$ and $\Pi_2$ in Lemma~\ref{theo:lemma_cuff_theo1}, and recalling that we replace $\Pi_1$ and $\Pi_2$ by $E_1$ and $E_2$, it follows that
\begin{align}
\min_{h_1\in\calH_1}\mathbf{E}[\tilde{d_1}(X,\hatX_1,h_1(U_1))]&=H(X|U_1),\label{DMSexpectconv:logloss_1}\\
\min_{h_2\in\calH_2}\mathbf{E}[\tilde{d_2}(X,\hatX_1,\hatX_2,h_2(U_1,U_2))]&=H(X|U_1,U_2).\label{DMSexpectconv:logloss_2}
\end{align}

Recall that the causal disclose satisfies that each $i\in[n]$, $W_{0,i}=X_i$  and $W_{1,i}=W_{2,i}$ equals to an arbitrary constant. Combining Lemma~\ref{theo:lemma_cuff_theo1}, Eq. \eqref{DMSexpectconv:logloss_1} and Eq. \eqref{DMSexpectconv:logloss_2} leads to
\begin{align}
\calS_{\mathrm{out}}=\bigcup_{\tilde{Q}\in\calB}&\left\{
\begin{array}{l}
(\vec{R_1},\vec{R_2},E_1,E_2):\\
R_1\geq I(X;U_1,V_1)\\
R_1+R_2\geq I(X;U_1,U_2,V_1,V_2)\\
r_1\geq I(X;V_1|U_1)\\
r_1+r_2\geq I(X;V_1,V_2|U_1,U_2)\\
D_1\geq\mathbf{E}\big[d_1(X,\hatX_1)\big]\\
D_2\geq\mathbf{E}\big[d_2(X,\hatX_2)\big]\\
E_1\leq H(X|U_1)\\
E_2\leq H(X|U_1,U_2)
\end{array}
\right\}. \label{DMSexpectconv:region_invloking_logloss}
\end{align}

To complete the proof of Lemma~\ref{theo:lemma_DMS_expect_corollary5}, we need to show that $\calS_{\mathrm{out}}=\calR_{\mathrm{out}}$. To do so, let $(\vec{R_1},\vec{R_2},D_1,D_2,E_1,E_2)\subseteq\calS_{\mathrm{out}}$. Furthermore, let $\hatX_1'\triangleq(U_1,V_1)$. It follows that
\begin{align}
R_1&\geq I(X;U_1,V_1)\\
&=I(X;\hatX_1),\\
E_1&\leq H(X|U_1)\\
&=H(X|U_1,V_1)+I(X;V_1|U_1)\\
&= H(X)-\big(I(X;U_1,V_1)-I(X;V_1|U_1)\big)\\
&= H(X)-\big\{I(X;U_1,V_1)-I(X;V_1|U_1)\big\}^+\label{DMSexpectconv:lemma_calS_adding+}\\
&\leq H(X)-\big\{I(X;\hatX_1')-r_1\big\}^+. \label{DMSexpectconv:lemma_calS_r1}
\end{align}

Further defining $(\hatX_1',\hatX_2')\triangleq(U_1,U_2,V_1,V_2)$, it follows that
\begin{align}
R_1+R_2&\geq I(X;U_1,U_2,V_1,V_2)\\
&=I(X;\hatX_1,\hatX_2),\\
E_2&\leq H(X|U_1,U_2)\\
&=H(X|U_1,U_2,V_1,V_2)+I(X;V_1,V_2|U_1,U_2)\\
&=H(X)-\big(I(X;U_1,U_2,V_1,V_2) -I(X;V_1,V_2|U_1,U_2)\big)\\
&=H(X)-\big\{I(X;U_1,U_2,V_1,V_2) -I(X;V_1,V_2|U_1,U_2)\big\}^+\\
&\leq H(X)-\big\{I(X;\hatX_1',\hatX_2')-r_1-r_2\big\}^+,
\end{align}
which implies $(\vec{R_1},\vec{R_2},D_1,D_2,E_1,E_2)\subseteq\calR$. Thus, $\calS_{\mathrm{out}}\subseteq\calR_{\mathrm{out}}$.

To show $\calR_{\mathrm{out}}\subseteq\calS_{\mathrm{out}}$, consider $(\vec{R_1},\vec{R_2},D_1,D_2,E_1,E_2)\subseteq\calR_{\mathrm{out}}$. Define $(U_1',V_1')\triangleq\hatX_1$. The Markov chain $U_1'-\hatX_1-X$ holds and thus
\begin{align}
H(X|U_1')=H(X)-\big\{I(X;\hatX_1)-r_1\big\}^+\label{DMSexpectconv_lemma_calR_condition1}.
\end{align}
Note that \eqref{DMSexpectconv_lemma_calR_condition1} is possible since $H(X|\hatX_1)\leq H(X|U_1')\leq H(X)$ and the right side of Eq. \eqref{DMSexpectconv_lemma_calR_condition1} lies in the interval $[H(X|\hatX_1),H(X)]$. Using the definition of $\calR_{\mathrm{out}}$, it follows that
\begin{align}
R_1&\geq I(X;\hatX_1)\\
&=I(X;U_1',V_1'),\\
E_1&\leq H(X)-\big\{I(X;\hatX_1)-r_1\big\}^+\\
&=H(X|U_1'),\\
r_1&\geq H(X|U_1')-H(X|\hatX_1)\label{DMSexpectconv:lemma_calR_r1}\\
&=H(X|U_1')-H(X|U_1',V_1')\\
&=I(X;V_1'|U_1'),
\end{align}
where Eq. \eqref{DMSexpectconv:lemma_calR_r1} follows from Eq. \eqref{DMSexpectconv_lemma_calR_condition1}. We further define $(U_1',U_2',V_1',V_2')\triangleq(\hatX_1,\hatX_2)$. It follows that the Markov chain $(U_1',U_2')-(\hatX_1,\hatX_2)-X$ holds and thus
\begin{align}
H(X|U_1',U_2')=H(X)-\big\{I(X;\hatX_1,\hatX_2)-r_1-r_2\big\}^+\label{DMSexpectconv_lemma_calR_condition2}.
\end{align}
Note that \eqref{DMSexpectconv_lemma_calR_condition2} is possible since $H(X|\hatX_1,\hatX_2)\leq H(H(X|U_1',U_2')\leq H(X)$ due to the Markov chain and the right side of Eq. \eqref{DMSexpectconv_lemma_calR_condition2} lies in the interval $[H(X|\hatX_1,\hatX_2),H(X)]$.
Again, using the definition of $\calR_{\mathrm{out}}$, it follows that
\begin{align}
R_1+R_2&\geq I(X;\hatX_1,\hatX_2)\\
&=I(X;U_1',U_2',V_1',V_2'),\\
E_2&\leq H(X)-\big\{I(X;\hatX_1,\hatX_2)-r_1-r_2\big\}^+\\
&=H(X|U_1',U_2'),\\
r_1+r_2&\geq H(X|U_1',U_2')-H(X|\hatX_1,\hatX_2)\label{DMSexpectconv:lemma_calR_r1r2}\\
&=H(X|U_1',U_2')-H(X|U_1',U_2',V_1',V_2')\\
&=I(X;V_1',V_2'|U_1',U_2'),
\end{align}
where Eq. \eqref{DMSexpectconv:lemma_calR_r1r2} follows from Eq. \eqref{DMSexpectconv_lemma_calR_condition2}. The above equations implies that $(\vec{R_1},\vec{R_2},D_1,D_2,E_1,E_2)\subseteq\calS_{\mathrm{out}}$ and thus $\calR_{\mathrm{out}}\subseteq\calS_{\mathrm{out}}$. The proof of Lemma \ref{theo:lemma_DMS_expect_corollary5} is completed by noting that $\calS_{\mathrm{out}}=\calR_{\mathrm{out}}$.

\subsubsection{The Rationality of Causal Disclosure}\label{sec:Conv_DMS_expectation_why}
We next explain why such an assumption is reasonable for the converse of Theorem~\ref{theo:DMS_expectation}.  As shown in Section~\ref{sec:Conv_DMS_expectation_lemma_3}, we bound the distortion of Eve denoted by $E_1$ and $E_2$ by using log-loss distortion as payoff functions and choosing the causal disclosure such that for each $i\in[n]$, $W_{0,i}=X_i$  and $W_{1,i}=W_{2,i}$ equals to an arbitrary constant. Under logloss distortion, Eve's payoff functions equal to equivocation terms $H(X^n|M_1)$ and $H(X^n|M_1,M_2)$. Recall that the causal disclosure is added to the successive refinement problem (cf. Fig. \ref{fig:causal_disclosure}). It follows that the public messages $(M_1,M_2)$ observed by Eve and the source sequence $X^n$ under the causal disclosure setting is exactly the same as the original problem without the causal disclosure. Thus, the normalized equivocation given by Eq. \eqref{equi:1} and \eqref{equi:2} is the exact normalized equivocation for Eve without the causal disclosure. Using the relationship between mutual information and equivocation that $I(X^n;M_1)=H(X^n)-H(X^n|M_1)$ and $I(X^n;M_1,M_2)=H(X^n)-H(X^n|M_1,M_2)$ and noting that maximal leakage is lower bounded by mutual information, the converse of Theorem~\ref{theo:DMS_expectation} is completed.

\section{Conclusion}\label{sec:conclusion}
We studied the successive refinement setting of Shannon cipher system that models multiuser secure communication with secret keys over a noiseless channel. Under both JEP and expected distortion reliability constraints, we derived inner and outer bounds for the asymptotic normalized information leakage region measured via maximal leakage for DMS under bounded distortion measures. Our bounds match under mild conditions on the key rates, which correspond to partial secrecy. Our result revealed the fundamental trade-off between reliability and secrecy. Counter-intuitively, although JEP appears a stronger reliability constraint, the leakage region under JEP is identical to the corresponding region under expected distortion for certain sources. To prove the converse result under expected distortion, we study a causal disclosure setting, where the eavesdropper could additionally acquire the past source and reconstruction symbols. With proper specialization, the converse result for the case with causal disclosure yields the desired converse result in Theorem \ref{theo:DMS_expectation}. It is possible that the converse proof of Theorem~\ref{theo:DMS_expectation} might be simplified without considering causal disclosure. However, given the generality and potential applications of the causal disclosure setting, we believe our converse proof is of independent interest. Specifically, a critical step in the converse proof Theorem~\ref{theo:DMS_expectation} is Lemma~\ref{theo:lemma_DMS_expect_corollary5}, which is itself of interest. This is because Lemma~\ref{theo:lemma_DMS_expect_corollary5} implies that with causal disclosure, the information leakage to Eve is not increased. In other words, having additional access to past source symbols to estimate the current source symbol cannot help Eve to have better performance. Such a result is in stark contrast to lossy source coding with causal side information, where the causally available side information can strictly reduce the compression rate and allow better performance~\cite[Example 11.1]{el2011network}.

There are several avenues for future research. Firstly, it is of interest to generalize our results to other multiterminal lossy source coding problems, e.g., Gray-Wyner~\cite{gray1974source} and multiple descriptions~\cite{ahlswede1986multiple}, and uncover the reliability-leakage tradeoff for more diverse multiuser secure communication settings with secret keys. Secondly, it is worthwhile to generalize our results to the noisy lossy source coding setting~\cite{wolf1970noisy}, where the source sequence to be compressed is available indirectly via a noisy channel. Such a setting is practical in certain applications and could be related to semantic compression~\cite{liu-vincent2021isit}. Thirdly, one can also adopt other secrecy metric beyond maximal leakage, e.g., maximal $\alpha$-leakage~\cite{liao2019alpha} and maximal $(\alpha,\beta)$-leakage~\cite{Gilani2022alphabeta}. Finally, it is worthwhile to consider a perception constraint~\cite{blau2018perception,chen2022perception} and investigate the tradeoff among reliability, perception and secrecy to further understand privacy constrained efficient compression of images and videos.

\section*{ACKNOWLEDGMENT}
The authors would like to thank the Associate Editor and anonymous reviewers for helpful comments and suggestions, which helps improve the quality of the present manuscript.

\bibliographystyle{IEEEtran}
\bibliography{IEEEfull_fei}

\end{document}